\documentclass[preprint,10pt]{article}
\usepackage[a4paper, total={6in, 10in}]{geometry}

\usepackage{array}
\usepackage{graphicx}
\usepackage{url}
\usepackage{setspace}
\usepackage{amssymb}
\usepackage{amsmath}
\usepackage{multirow}
\usepackage{mathtools}
\usepackage{amsthm}
\usepackage{hyperref}

\newtheorem{definition}{Definition}
\newtheorem{lemma}{Lemma}
\newtheorem{theorem}{Theorem}
\newtheorem{corollary}{Corollary}

\begin{document}

\title{Differential Privacy Via a Truncated and Normalized Laplace Mechanism \footnote{\textbf{This is a pre-print of an article published in Journal of Computer Science and Technology. The final authenticated version is available online at: \url{https://doi.org/10.1007/s11390-020-0193-z}.}}}

\author{William Lee Croft \and J\"{o}rg-R\"{u}diger Sack \and Wei Shi}



\maketitle

\begin{abstract}
	When querying databases containing sensitive information, the privacy of individuals stored in the database has to be guaranteed. Such guarantees are provided by differentially private mechanisms which add controlled noise to the query responses. However, most such mechanisms do not take into consideration the valid range of the query being posed. Thus, noisy responses that fall outside of this range may potentially be produced. To rectify this and therefore improve the utility of the mechanism, the commonly used Laplace distribution can be truncated to the valid range of the query and then normalized. However, such a data-dependent operation of normalization leaks additional information about the true query response thereby violating the differential privacy guarantee.
		
	Here, we propose a new method which preserves the differential privacy guarantee through a careful determination of an appropriate scaling parameter for the Laplace distribution. We also generalize the privacy guarantee in the context of the Laplace distribution to account for data-dependent normalization factors and study this guarantee for different classes of range constraint configurations. We provide derivations of the optimal scaling parameter (i.e., the minimal value that preserves differential privacy) for each class or provide an approximation thereof. As a consequence of this work, one can use the Laplace distribution to answer queries in a range-adherent and differentially private manner.
\end{abstract}

\section{Introduction}

Since its introduction in 2006, differential privacy \cite{2} has become one of the most well-studied disclosure control methods that offer a formal privacy guarantee for the release of query responses for sensitive databases. At a high level, it prevents accurate inferences about the contents of the database by using a randomization mechanism to add noise to the query responses. By carefully controlling the noise, a data custodian can protect the privacy of individuals in the database while allowing for aggregate level analysis to be effectively carried out.

In many applications, queries will have natural and publicly known constraints on their range. For example, the percentage of satisfied (or dissatisfied) customers in a survey must fall in the range of 0 to 100. Inferences about the range of a query can often be drawn with ease. The sum of positive valued attributes cannot be less than zero and the number of records in a database that satisfy a particular predicate cannot be greater than the size of the database. If a valid range is known for a particular attribute, a valid range for most queries posed over that attribute can be calculated. In practice, most differentially private mechanisms add noise without consideration for such constraints. This implies that noisy responses may be generated outside of the valid range of the query.

The generation of out-of-bounds noisy query responses can be detrimental both in terms of downstream compatibility of the responses with other software and in terms of the ability to perform useful analysis on the database. To avoid this, the noisy responses should conform to the true range of the query if it is public information. A boundary-snapping process which involves snapping out-of-bounds values to the nearest valid value is sometimes used in practice, however, this does not take into consideration the impact of this operation on the overall utility of the noisy responses. To achieve a low level of expected distortion in noisy query responses, probability mass should ideally be focused near the true query response as much as possible. Boundary-snapping leads to pooling of probability mass at the boundaries of the valid range, which may not be near the true query response.

We provide an alternative to boundary-snapping where we first truncate the probability density function (PDF) used for the generation of noisy query responses to the valid range of the query and then normalize it. Subsequent normalization is required to restore the truncated function to a valid PDF. Since normalization involves the multiplication of the function by a factor greater than 1 and most PDFs have higher probability density around the location parameter, this process leads to greater increases in probability density for noisy responses nearby the true response compared to those farther away. This is beneficial for the utility of the mechanism as it leads to a lower level of expected distortion when noise is drawn from the new distribution.

The method of truncation and normalization has been previously applied \cite{11} in the context of the Laplace distribution, a very widely-used distribution in differential privacy. However, that study did not account for the fact that the normalization factor is a function of the true query response, which is sensitive information. Due to this, normalization in fact leaks further information, and additional modification to the PDF is required to account for this. In this work, we study this leakage of information and address it by proposing new versions of truncated and normalized PDFs from a Laplace distribution which, as we show, preserve the differential privacy guarantee.

\subsection{Contributions}

Our work focuses on the design of differentially private mechanisms which adhere to query range constraints via the truncation and normalization of the PDF from a Laplace distribution. More specifically, our main contributions are:

\begin{itemize}
	
	\item We demonstrate that the process of truncation and normalization of a Laplace PDF is not sufficient on its own to preserve the differential privacy guarantee. The main reason for this results from the data-dependent nature of the normalization.
	
	\item We show how this can be corrected by carefully calculating the scaling parameter for the Laplace distribution. For this, we generalize the differential privacy guarantee in the context of a Laplace mechanism to incorporate data-dependent normalization factors. 
	
	\item We use our generalized privacy guarantee to guide a study of range-adherent mechanisms with respect to different classes of range constraints they are able to adhere to. For each class we show how to derive an optimal scaling parameter or an approximation thereof.

\end{itemize}

\section{Literature Review}

Privacy-preserving analysis of sensitive data has a long history as can be witnessed by the many different research directions that have been investigated; for a survey of early such works, see \cite{1}. In recent years, the framework of differential privacy \cite{15} has been gaining a great deal of traction as a preferred method for analysis of sensitive databases. By adding noise drawn from an appropriately configured Laplace distribution to database query responses prior to their release, it has been proven that the ability of attackers to distinguish between potential configurations of database records can be limited \cite{2}. This concept has subsequently spawned diverse studies including the number of queries that can be safely answered, the utility of the noisy responses, and variations on how to add noise; see e.g., \cite{3,4}.

Some studies on preserving utility when adding noise have focused on publicly known constraints. In these, the goal for the generation of noisy responses is to respect the constraints and thereby provide better utility. In the non-interactive setting, where batches of related queries are posed, known relationships between the responses to a batch of queries should be preserved. This was first considered in the context of the queries needed to break a contingency table for a database up into marginals (i.e., projections of record counts over subsets of the database attributes). Linear programming was applied after a transformation to the Fourier domain to ensure consistency such that the noisy marginals don't appear to have been derived from different databases \cite{5}. Subsequent work has improved further upon this concept by applying post-processing to vectors of noisy query responses in order to achieve consistency between the noisy values while minimizing the distance between the original noisy vector and the post-processed version \cite{6,7}.

In the interactive setting in differential privacy independent queries are answered on an individual basis. There, it may be desirable to ensure that noisy responses are consistent with a (known) valid range of the query. A well-known implementation of this is the truncated geometric mechanism \cite{8}, a variant of the geometric mechanism \cite{9} which snaps out-of-bounds values to the nearest valid response. In many cases, this approach can lead to a high density of noisy responses on the boundaries of the valid range of the query. To avoid undesirable properties such as these types of spikes in probability density, the explicit fair mechanism \cite{10} was designed with to answer queries in a range-adherent manner while providing properties such as a monotonically decaying PDF. Both the geometric and the explicit fair mechanisms are specifically designed to answer counting queries (i.e., the number of database records satisfying a particular predicate). As such, they only produce integer-valued responses and cannot be applied to more general numeric queries.

The concept of snapping out-of-bounds noisy responses to the nearest valid value can be equally applied to any type of randomization mechanism. This has been studied in the context of the Laplace mechanism which produces real-valued output and can be applied to general numeric queries. Another variant of the Laplace mechanism involves truncation of the PDF used by a Laplace distribution to the valid range of the query followed by normalization \cite{11}. Although the method of normalization might be preferred for a greater inflation of probability mass around the true query response as opposed to the boundaries of the valid range, the study neglected to consider the impact of the information leakage inherent in the process of normalization. We note that other mechanisms using truncated PDFs have been studied as well, however, the truncation was not designed to match the valid range of the query and thus they cannot be applied as a means to achieve adherence to the valid range \cite{12,13}.

\section{Preliminaries for Truncation and Normalization}

Differential privacy \cite{15} (Definition \ref{def:dp}) enforces the guarantee that similar databases must have similar probabilities of producing the same noisy query responses. More specifically, it guarantees that for every pair of adjacent databases (i.e., databases that differ by a single record) and every noisy query response, the respective probabilities of the two databases to produce the given noisy query response must be within a factor of $e^{\epsilon}$ of each other, where $\epsilon$ is a privacy parameter specified by the data custodian. The lower the value of $\epsilon$ is set, the stricter the guarantee becomes as the distributions over the noisy query responses are forced to become more similar to each other.

\begin{definition} \label{def:dp}
	Let $\mathbb{D}$ be the set of potential database configurations, $f: \mathbb{D} \rightarrow \mathbb{R}$ be a query function and $K: \mathbb{R} \rightarrow \mathbb{R}$ be a randomization mechanism. For $K$ to satisfy the \textbf{differential privacy guarantee} \cite{15}, the following condition must hold for all pairs of adjacent databases $D_1, D_2 \in \mathbb{D}$:
	
	\begin{equation} \label{eq:guar}
	\Pr \left(  K \left( f \left( D_1 \right) \right) = r \right) \leq e^{\epsilon} \Pr \left( K \left( f \left( D_2 \right) \right) = r \right) \hspace{0.3em} \forall r \in \mathbb{R}
	\end{equation}
\end{definition}

For answering general numeric queries in a differentially private manner, the Laplace distribution is commonly used to implement the randomization mechanism.
	
\begin{definition} \label{def:lap}
	The \textbf{Laplace distribution} for a continuous random variable $x$, configured by a location parameter $\mu$ and a scaling parameter $\sigma$ is given by:
	
	\begin{equation} \label{eq:laplace}
		Lap \left( x | \mu, \sigma \right) = \frac{e^{-\frac{\left| \mu - x \right|}{\sigma}}}{2 \sigma}.
	\end{equation}
	
\end{definition}
	
To satisfy the differential privacy guarantee, the scaling parameter must be set in terms of the chosen value of $\epsilon$ and the query sensitivity, $\Delta F$, which is the largest possible difference between the true query responses of a pair of adjacent databases. By setting the scaling parameter to $\frac{\Delta F}{\epsilon}$ and the location parameter to the true query response, values drawn from the resulting distribution can be used as the noisy query responses. This configuration is referred to as the Laplace mechanism \cite{15}.

When answering queries in a range-adherent manner, the PDF of a Laplace distribution can be truncated to the valid range of the query and then normalized. Since the probability density in the Laplace distribution exponentially decays as distance from the location parameter increases, multiplication by the normalization factor will induce greater increases in probability density for noisy query responses that are closer to the true response than those that are farther away. This is favourable for the utility of the mechanism as it is conducive to a low level of expected distortion in the noisy query responses. However, the operation of normalization is dependent upon the true query response and leaks sensitive information. Verifying the privacy guarantee becomes much more complex with this method and requires further modification of the PDF. We propose an alternative method (to that which is used by the standard Laplace mechanism) for the calculation of the scaling parameter in the Laplace distribution that accounts for the data-dependent normalization in order to preserve the privacy guarantee of differential privacy.

In this section, we provide the preliminaries that are needed to study the impact of truncation and normalization of a Laplace PDF with respect to the privacy guarantee. We first formalize our treatment of range constraints and lay out the calculations of normalization factors for classes of constraint configurations. We then examine the impact of adherence to range constraints and data-dependent normalization on the privacy guarantee. We define a generalized privacy guarantee to use in this setting. With these components laid out, we provide a small example to demonstrate the violation of the privacy guarantee when truncation and normalization are applied without further modification to the PDF.

In Table \ref{table:definitions}, we provide a quick reference for the definitions of the variables used throughout the paper.

\begin{table*}
	\small
	\caption{Definitions of variables}
	\begin{tabular}{ | p{1.3cm} | p{11.6cm} | } \hline
		$\mathbb{D}$ & The set of potential database configurations\\ \hline
		$K$ & A function $K:\mathbb{R} \rightarrow \mathbb{R}$ representing a randomization mechanism \\ \hline
		$f$ & A function $f:\mathbb{D} \rightarrow \mathbb{R}$ representing a query posed on a database\\ \hline
		$\Delta F$ & Query sensitivity - The maximum possible difference between the query responses of two adjacent databases. This is a non-zero, positive, real value. \\ \hline
		$\epsilon$ & Privacy parameter - The user-specified, desired level of privacy. This is a non-zero, positive, real value. \\ \hline
		$\sigma$ & Scaling parameter (for cases where all PDFs share the same scaling parameter) - A value calculated to determine the scale of noise required to release a database query. This is a non-zero, positive, real value. \\ \hline
		$x$ & The continuous random variable used for a PDF \\ \hline
		$n_j$ & The normalization factor required for database $D_j$ \\ \hline
		$\sigma_j$ & The scaling parameter required for database $D_j$ \\ \hline
		$\Delta x_j$ & The distance between $f(D_j)$ and the continuous random variable $x$ \\ \hline
		$CL_j$ & The set of finite constraints to the left of $f(D_j)$ \\ \hline
		$CR_j$ & The set of finite constraints to the right of $f(D_j)$ \\ \hline
		$I_L$ & A boolean variable used to indicate the presence of a constraint spanning to negative infinity \\ \hline
		$I_R$ & A boolean variable used to indicate the presence of a constraint spanning to infinity \\ \hline
		$\Delta R_{j}$ & The distance between $f(D_j)$ and the right endpoint of a constraint spanning to negative infinity \\ \hline
		$\Delta L_{j}$ & The distance between $f(D_j)$ and the left endpoint of a constraint spanning to infinity \\ \hline
		$\Delta R_{jk}$ & The distance between $f(D_j)$ and the right endpoint of a constraint $C_k$ \\ \hline
		$\Delta L_{jk}$ & The distance between $f(D_j)$ and the left endpoint of a constraint $C_k$ \\ \hline
		$L_j$ & A summation of the integrals of all constraints to the left of $f(D_j)$ \\ \hline
		$R_j$ & A summation of the integrals of all constraints to the right of $f(D_j)$ \\ \hline
		$i$ & A real-valued variable used to indicate the number of spans of $\Delta F$ between two databases \\ \hline		
	\end{tabular}
	\label{table:definitions}
\end{table*}

\subsection{Constraint Configurations}

Given constraints on the range of a query represented as a set of ranges within which query responses cannot occur, the goal is to truncate a Laplace PDF to remove these constrained ranges, calculate an appropriate normalization factor to restore the integral over the truncated span(s) to 1, and calculate an appropriate scaling parameter. The required scaling parameter is dependent upon a privacy parameter $\epsilon$ specified by the data custodian and the position of the location parameter (i.e., the true query response) relative to the constraints.

The scaling parameter must be selected such that the differential privacy guarantee is satisfied. The higher the scaling parameter is set, the farther the probability mass becomes spread from the location parameter. Therefore, calculating the lowest possible scaling parameter that satisfies the privacy guarantee is one of the core aspects of our work. We consider the selection of the lowest possible scaling parameter to be the criterion for optimal utility within the class of mechanisms using a truncated and normalized Laplace PDF.

Allowing for truncated spans in a PDF provides a natural way to incorporate range constraints into a randomization mechanism. Given the wide variety of queries that can be posed on a database, many different configurations of constraints may arise in practice. To simplify the work required in designing appropriate mechanisms for the various constraint configurations, we categorize these configurations into classes such that for each class, a corresponding mechanism can be developed. Visual representations of possible constraint configurations from each class are shown in Figure \ref{fig:constraints}.

The first class is characterized by a single constraint that begins at any point and spans either to infinity or to negative infinity (Figure \ref{fig:constraints}-a). An example of this is a query that has the set of natural numbers or positive real numbers as its range.

The second class may have any number of finite spans of constraints (Figure \ref{fig:constraints}-b). These configurations would be typical of queries in which the range is subject to a periodic function. For example, a query may be posed about dates on which individuals visited a particular medical clinic where it is known that the clinic is closed during certain periods of the year.

We also consider a final class of arbitrary constraint configurations to handle any combination of the first two classes (Figure \ref{fig:constraints}-c). Common instantiations of this class are queries that have a single finite span of valid responses. For example, a query may have a single finite range of possible age values.

\begin{figure}
	\centering
	\includegraphics[width=0.4\textwidth]{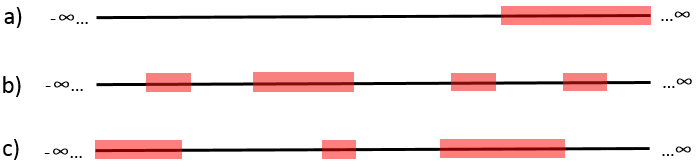}
	\caption{An example from each class of constraint configurations is shown. The black line represents a real-valued, one-dimensional range and the shaded red blocks represent constraints. Constraints shown on the endpoint of the line span infinitely in that direction. The constraints classes are: a) a single infinite constraint, b) arbitrary finite constraints, and c) arbitrary constraints.}
	\label{fig:constraints}
\end{figure}

\subsection{Normalization Factors}

Since constraints represent infeasible spans in the range of the query, the PDF must be truncated to remove these spans. This has the effect of producing a piece-wise function whose integral is no longer equal to 1. In order to restore this to a PDF, it must be scaled by a normalization factor such that the integral over the valid space is equal to 1. The normalization factor is calculated as the reciprocal of the integral of the truncated PDF. This integral is calculated as the sum of the integrals of the removed spans of the PDF subtracted from 1. An example of a truncated and normalized Laplace PDF is shown in Figure \ref{fig:normalization}.

\begin{figure} [ht]
	\centering
	\includegraphics[width=0.4\textwidth]{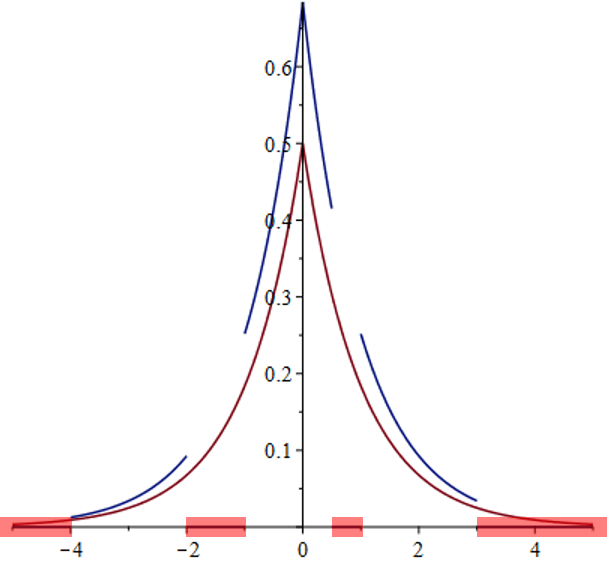}
	\caption{A graph of two Laplace PDFs with location parameters of 0 and scaling parameters of 1. The shaded blocks show an example of different types of constraints (infinite and finite on both the left and right of the location parameter). The red line shows the original Laplace PDF and the blue line shows one that has been truncated and then normalized.}
	\label{fig:normalization}
\end{figure}

Integrals of a Laplace PDF can be easily calculated using the cumulative distribution function (CDF).

\begin{definition} \label{def:cdf}
	The \textbf{cumulative distribution function} corresponding to a Laplace distribution of $Lap \left( x | \mu, \sigma \right)$ is given by:

	\begin{equation} \label{eq:laplace_cdf}
	CDF \left( x | \mu, \sigma \right) = \frac{1}{2} + \frac{1}{2} sgn \left( x - \mu \right) \left( 1 - e^{-\frac{\left| \mu - x \right|}{\sigma}} \right)
	\end{equation}

\end{definition}

We define the following notation for the representation of constraints. Let $CL_j$ be the set of finite constraints to the left of $f(D_j)$ for any database $D_j \in \mathbb{D}$. For any constraint $C_k \in CL_j$, let $\Delta L_{jk}$ and $\Delta R_{jk}$ be the distance between $f(D_j)$ and the left and right endpoints, respectively, of $C_k$. Let $I_L$ be an indicator variable set to either 0 or 1 to denote the absence or presence of a constraint that spans to negative infinity and let $\Delta R_j$ be the distance between the right endpoint of such a constraint, if it exists, and $f(D_j)$. These concepts are all defined analogously for a set $CR_j$ of finite constraints to the right of $f(D_j)$ with $I_R$ acting as the indicator variable for a constraint spanning to infinity and $\Delta L_j$ acting as the distance between the left endpoint of such a constraint, if it exists, and $f(D_j)$. Finally, let $\sigma_j$ be a data-dependent scaling parameter to be used when querying a database $D_j$.

Let $L_j$, defined in Formula \ref{eq:integral_L}, be the sum of all integrals of the constraints to the left of $f(D_j)$. The integral of a finite constraint to the left of $f(D_j)$ can be calculated as the CDF at the right endpoint of the constraint minus the CDF at the left endpoint (shown as the expression inside the summation of Formula \eqref{eq:integral_L}). The integral of a constraint that spans to negative infinity can be calculated as the CDF at the right endpoint of the constraint (shown in the term outside the summation of Formula \eqref{eq:integral_L}). 

{
	\begin{equation} \label{eq:integral_L}
	L_j = 
	\sum_{C_k \in CL_j} \left( \frac{e^{-\frac{\Delta R_{jk}}{\sigma_j}}-e^{-\frac{\Delta L_{jk}}{\sigma_j}}}{2} \right) + 
	I_L \frac{e^{-\frac{\Delta R_j}{\sigma_j}}}{2}
	\end{equation}
}

Analogous calculations for the integrals of finite constraints to the right of $f(D_j)$ and a constraint spanning to infinity are shown in Formula \eqref{eq:integral_R}, where $R_j$ is the sum of all integrals of the constraints to the right of $f(D_j)$.

{
	\begin{equation} \label{eq:integral_R}
	R_j = 
	\sum_{C_k \in CR_j} \left( \frac{e^{-\frac{\Delta L_{jk}}{\sigma_j}}-e^{-\frac{\Delta R_{jk}}{\sigma_j}}}{2} \right) +
	I_R \frac{e^{-\frac{\Delta L_j}{\sigma_j}}}{2}
	\end{equation}
}

The total integral of the removed space is simply the summation of all constraint integrals (the sum of Formulae \eqref{eq:integral_L} and \eqref{eq:integral_R}). The integral of the truncated PDF is the sum of the constraint integrals subtracted from 1. The normalization factor $n_j$ is the reciprocal of the integral of the truncated PDF, as shown in Formula \eqref{eq:norm}.

{
	\begin{equation} \label{eq:norm}
	n_j = \frac{1}{1 - \left( L_j + R_j \right)}
	\end{equation}
}

\subsection{Generalized Differential Privacy Guarantee}

When using the basic Laplace distribution to implement the randomization mechanism, the privacy guarantee of Definition \ref{def:dp} can be simplified to the form shown in Formula \eqref{eq:reg_guar} for all pairs of adjacent databases $D_1, D_2 \in \mathbb{D}$, where $\Delta x_j$ is the distance between between the continuous random variable $x$ and $f(D_j)$ for any database $D_j \in \mathbb{D}$.

{
	\begin{equation} \label{eq:reg_guar}
	e^{\frac{\Delta x_2}{\sigma}} \leq e^\epsilon e^{\frac{\Delta x_1}{\sigma}} \hspace{0.3em} \forall x \in \mathbb{R}
	\end{equation}
}

When using normalization factors induced by constraints and allowing for differing scaling parameters, we obtain a new general form of the privacy guarantee for all pairs of adjacent databases $D_1, D_2 \in \mathbb{D}$ shown in Formula \eqref{eq:guar_simp}.

{
	\begin{equation} \label{eq:guar_simp}
	\frac{\sigma_2}{\sigma_1}
	\left( \frac{n_1}{n_2} \right)
	\left( \frac{e^{\frac{\Delta x_2}{\sigma_2}}}{e^{\frac{\Delta x_1}{\sigma_1}}} \right)
	\leq e^\epsilon
	\hspace{0.3em} \forall x \in \mathbb{R}
	\end{equation}
}

To simplify the analysis, we assume that when two databases, $D_1, D_2 \in \mathbb{D}$ are being compared in a privacy guarantee, it will always be the case that $f(D_1) \ge f(D_2)$. As a result, a form for the symmetric case of the guarantee (i.e., where $f(D_1) \le f(D_2)$) must also be considered, as shown in Formula \eqref{eq:guar_sym}. To assert the differential privacy guarantee, both forms must be satisfied for all pairs of adjacent databases that apply to their respective cases.

{
	\begin{equation} \label{eq:guar_sym}
	\frac{\sigma_1}{\sigma_2}
	\left( \frac{n_2}{n_1} \right)
	\left( \frac{e^{\frac{\Delta x_1}{\sigma_1}}}{e^{\frac{\Delta x_2}{\sigma_2}}} \right)
	\leq e^\epsilon
	\hspace{0.3em} \forall x \in \mathbb{R}
	\end{equation}
}

\subsection{Guarantee Generalization for Arbitrary Distances} \label{sec:gen_dist}

The differential privacy guarantee is generally interpreted in terms of pairs of adjacent databases. For such databases, the distance between their true query responses is upper bounded by the query sensitivity $\Delta F$. Although not typically done, the privacy guarantee can also be interpreted in terms of the actual distance between query responses of arbitrary databases (i.e, pairs that are not necessarily adjacent). This is a useful interpretation for our setting since the distance between the databases can be used as a means to calculate the difference between their normalization factors as well.

For a pair of databases $D_1, D_2 \in \mathbb{D}$ having PDFs with location parameters in the same span of truncated space, the constraint endpoint distances of $D_2$ can be written in terms of the variables used for $D_1$. This can be done by applying an offset equal to the distance between the two location parameters. We define a distance variable $i$, representing the number of spans of $\Delta F$ between the location parameters, which we utilize to create a new representation of the constraint integral calculations. Since we assume that $f(D_1) \geq f(D_2)$, we know that $f(D_2)$ will be closer to all constraints to the left and farther from all constraints to the right compared to $f(D_1)$. We therefore use an offset of $-i \Delta F$ for distances to the left and an offset of $i \Delta F$ for distances to the right. The distance offset versions of the constraint integral calculations are shown in Formulae \eqref{eq:integral_Lj} and \eqref{eq:integral_Rj}.

{
	\begin{equation} \label{eq:integral_Lj}
	L_2 = 
	\sum_{C_k \in CL_1} \left( \frac{e^{-\frac{\Delta R_{1k} - i \Delta F}{\sigma_2}}-e^{-\frac{\Delta L_{1k} - i \Delta F}{\sigma_2}}}{2} \right) + 
	I_L \frac{e^{-\frac{\Delta R_1 - i \Delta F}{\sigma_2}}}{2}
	\end{equation}
}

{
	\begin{equation} \label{eq:integral_Rj}
	R_2 = 
	\sum_{C_k \in CR_1} \left( \frac{e^{-\frac{\Delta L_{1k} + i \Delta F}{\sigma_2}}-e^{-\frac{\Delta R_{1k} + i \Delta F}{\sigma_2}}}{2} \right) +
	I_R \frac{e^{-\frac{\Delta L_1 + i \Delta F}{\sigma_2}}}{2}
	\end{equation}
}

To use these integral calculations in the privacy guarantee, we must consider how the required multiplicative bound between the probabilities of the databases producing the same noisy response should reflect this modification. Since $e^\epsilon$ is the multiplicative bound for adjacent databases (meaning their query responses are separated by a distance of at most $\Delta F$), it follows that databases whose query responses are separated by $i$ spans of $\Delta F$ should have a multiplicative bound given by the product of $e^\epsilon$ multiplied by itself $i$ times. In other words, the multiplicative bound must become $e^{i \epsilon}$.

In fact, $i$ does not need to be an integer; it can be real-valued. The important requirement to maintain in this generalization is that for any two databases separated by $\Delta F$, the multiplicative difference in their probabilities will be bounded by $e^\epsilon$. To show this, consider two databases separated by $i$ spans of $\Delta F$. Since $\frac{1}{i}$ times that distance is equal to $\Delta F$, we must show that the multiplicative bound for $\frac{1}{i}$ of the total distance is $e^\epsilon$. This is easily shown as follows:

{
	\begin{equation} \label{eq:dist_proof}
	{e^{i \epsilon}}^{\frac{1}{i}} = e^\epsilon
	\end{equation}
}


To apply this generalization to the privacy guarantee, we use $L_1$ and $R_1$ defined according to Formulae \eqref{eq:integral_L} and \eqref{eq:integral_R} as well as $L_2$ and $R_2$ defined according to Formulae \eqref{eq:integral_Lj} and \eqref{eq:integral_Rj} to produce the generalized form of the privacy guarantee shown in Formula \eqref{eq:guar_dist}.

{
	\begin{equation} \label{eq:guar_dist}
	\frac{\sigma_2}{\sigma_1} 
	\left( \frac{1 - \left( L_2 + R_2 \right) }
	{1 - \left( L_1 + R_1 \right) } \right)
	\left( \frac{e^{\frac{\Delta x_2}{\sigma_2}}}{e^{\frac{\Delta x_1}{\sigma_1}}} \right)
	\leq e^{i \epsilon}
	\hspace{0.3em} \forall x \in \mathbb{R}
	\end{equation}
}

\subsection{Continuous Random Variable Worst Case Analysis} \label{sec:cont_var}

In the generalized privacy guarantee of Formula \eqref{eq:guar_dist}, we can see that the variables related to the continuous random variable of the PDFs appear only in the third fraction of the left-hand side of the inequality. We now look for the worst case with respect to the continuous random variable placement in order to eliminate the need for explicit quantification over this variable in the privacy guarantee. Since the left-hand side of the guarantee must always be less than or equal to the right-hand side, the worst case occurs when the fraction is maximized. 

Applying identities, the fraction containing the continuous random variable can be re-written as:

{
	\begin{equation} \label{eq:cont_var}
	e^{\frac{\Delta x_2}{\sigma_2}-\frac{\Delta x_1}{\sigma_1}}
	\end{equation}
}

This expression is maximized when the exponent is maximized. If we redefine the $\Delta x_2$ in terms of $\Delta x_1$ based on the distance between the location parameters, we can write out three different cases which cover the possible placements of the continuous random variable. These cases are shown in Formula \eqref{eq:cont_var_placement}. Recall that $f(D_1) > f(D_2)$ and $f(D_1) - f(D_2) = i \Delta F$.

\begin{equation} \label{eq:cont_var_placement}
	e^{\frac{\Delta x_2}{\sigma_2}-\frac{\Delta x_1}{\sigma_1}} = 
	\begin{cases}
		e^{\frac{\Delta x_1 - i \Delta F}{\sigma_2}-\frac{\Delta x_1}{\sigma_1}}  & x \le f(D_2) \\
		e^{\frac{i \Delta F - \Delta x_1}{\sigma_2}-\frac{\Delta x_1}{\sigma_1}} & f(D_2) < x < f(D_1) \\
		e^{\frac{\Delta x_1 + i \Delta F}{\sigma_2}-\frac{\Delta x_1}{\sigma_1}} & f(D_1) \le x \\
	\end{cases}
\end{equation}

Of these cases, the third will produce the highest value, thus the continuous random variable should be placed somewhere in the span of truncated space from $f(D_1)$ up to infinity. In order to determine which position within this space maximizes the value of the expression, we must consider the relationship between $\sigma_1$ and $\sigma_2$.

If $\sigma_1$ were always equal to $\sigma_2$ then the placement of the continuous random variable within the identified span would have no impact on the value of the expression. However, we would ideally like to take advantage of the fact that the farther a location parameter is from a constraint, the less of the PDF integral is lost during truncation. A larger remaining integral means that a lower normalization factor is needed. As we later show, higher scaling parameters are needed to compensate for the increased normalization factors. However, higher scaling parameters lead to higher levels of expected noise. Therefore, it is desirable to allow the scaling parameters to decrease as the normalization factors decrease.

We first consider the case of a single infinitely spanning constraint. If the constraint spans to infinity, this implies that in this configuration, $\sigma_1 \geq \sigma_2$. A greater $\sigma_1$ value means that the value of the expression goes up at $\Delta x_1$ increases. We must therefore work with the highest possible $\Delta x_1$ value. If the guarantee can be satisfied for this value then any other possible $\Delta x_1$ value will cause the value to the left-hand side to decrease, meaning that the guarantee will hold. We therefore place $x$ right on the border of the constraint going to infinity (if one exists). This means that $\Delta x_1$ is equal to $\Delta L_1$. The fraction can be updated to:

\begin{equation} \label{eq:cont_var_worst}
\scalebox{1.5}{$
	\frac{e^{-\frac{\Delta L_1}{\sigma_1}}}{e^{-\frac{\Delta L_1 + i \Delta F}{\sigma_2}}}
	$}
\end{equation}

While similar analysis could be provided for the case of a single constraint spanning to negative infinity, we do not explicitly consider this case as it can be easily treated as the case of a constraint spanning to infinity after a horizontal reflection has been applied.

In constraint configurations where there are no infinite constraints, the worst-case with differing scaling parameters is undefined since the continuous random variable can always be placed farther away. In such a configuration, it becomes necessary to use the same scaling parameter everywhere to avoid this problem. When this occurs, the fraction can be written as:

{
	\begin{equation} \label{eq:cont_var_same_sigma}
	e^{\frac{i \Delta F}{\sigma}}
	\end{equation}
}

We must also consider the form for the symmetric case of the privacy guarantee. This time, the worst-case occurs when we maximize:

{
	\begin{equation} \label{eq:cont_var_sym}
	e^{\frac{\Delta x_1}{\sigma_1}-\frac{\Delta x_2}{\sigma_2}}
	\end{equation}
}

When redefining the $\Delta x_2$ in terms of $\Delta x_1$ based on the distance between the location parameters, the three possible cases are as shown in Formula \eqref{eq:cont_var_placement_sym}.

\begin{equation} \label{eq:cont_var_placement_sym}
	e^{\frac{\Delta x_1}{\sigma_1}-\frac{\Delta x_2}{\sigma_2}} = 
	\begin{cases}
		e^{\frac{\Delta x_1}{\sigma_1}-\frac{\Delta x_1 - i \Delta F}{\sigma_2}} & x \le f(D_2) \\
		e^{\frac{\Delta x_1}{\sigma_1}-\frac{i \Delta F - \Delta x_1}{\sigma_2}} & f(D_2) < x < f(D_1) \\
		e^{\frac{\Delta x_1}{\sigma_1}-\frac{\Delta x_1 + i \Delta F}{\sigma_2}} & f(D_1) \le x \\
	\end{cases}
\end{equation}

The first and second cases will be larger than the third. The worst-case occurs when the second term of the exponent is minimized. The smallest possible value it can take on is 0, which occurs when $\Delta x_1 = i \Delta F$. The fraction can therefore be re-written as:

{
	\begin{equation} \label{eq:cont_var_sym_worst}
	e^{\frac{i \Delta F}{\sigma_1}}
	\end{equation}
}

For cases where the same $\sigma$ values are used due to the worst-case analysis from the symmetric form, this would be:

{
	\begin{equation} \label{eq:cont_var_sym_same_sigma}
	e^{\frac{i \Delta F}{\sigma}}
	\end{equation}
}

\subsection{Motivating Example} \label{sec:laplace_motiv}

We conclude this section with an example to illustrate the need for an alternate calculation of the scaling parameter when the Laplace PDF is truncated and normalized. Consider a case with a single constraint spanning to negative infinity where a fixed scaling parameter is used regardless of the position of the location parameter. The privacy guarantee would be written as shown in Formula \eqref{eq:motiv_example}.

{
	\begin{equation} \label{eq:motiv_example} 
	\frac{2-e^{-\frac{\Delta R_1 + i \Delta F}{\sigma}}}{2-e^{-\frac{\Delta R_1}{\sigma}}}  
	\left( e^{\frac{i \Delta F}{\sigma}} \right)
	\leq e^{i \epsilon}
	\end{equation}
}

The regular calculation of the scaling parameter for a Laplace mechanism is to set $\sigma=\frac{\Delta F}{\epsilon}$. When making this substitution, we have the form shown in Formula \eqref{eq:motiv_example_2}.

{
	\begin{equation} \label{eq:motiv_example_2} 
	\frac{2e^{i \epsilon}-e^{-\frac{\epsilon \Delta R_1}{\Delta F}}}{2-e^{-\frac{\epsilon \Delta R_1}{\Delta F}}}
	\leq e^{i \epsilon}
	\end{equation}
}

Now we consider the comparison of a database with a location parameter on the border of the constraint to any other database. To do so, we set $\Delta R_1$ to 0 and simplify to produce the form shown in Formula \eqref{eq:motiv_example_3}.

{
	\begin{equation} \label{eq:motiv_example_3} 
	2e^{i \epsilon}-1
	\leq e^{i \epsilon}
	\end{equation}
}

We can see that for any $i>0$ (i.e., any pair of databases that do not share the same true query response), the privacy guarantee is not satisfied. This shows that the regular calculation of the scaling parameter cannot be applied when normalization has occurred. In the remainder of the paper, we show how to calculate appropriate scaling parameters to preserve the differential privacy guarantee. Furthermore, since different location parameters induce different normalizations factors, we allow for the scaling parameter to be calculated as a function of the position of the location parameter relative to the constraints in order to achieve lower scaling parameters and thus better utility.

\section{Constraint Configurations}

Having defined classes of range constraints and their corresponding normalization factors, we now study these classes in the context of our generalized privacy guarantee. This form of the privacy guarantee captures the information leakage induced by data-dependent normalization, allowing us to investigate how the selection of appropriate scaling parameters can be used to account for the information leakage. Throughout this section, we define an instantiation of the privacy guarantee for each class of constraint configurations according to the corresponding calculation of the normalization factors and the worst-case analysis of the continuous random variable. Through the manipulation of these instantiations, we derive the required scaling parameters. For conciseness within this section, the full proofs are provided in the appendix.

\subsection{Single Infinite Constraint}

We begin with the class of single infinitely spanning constraints. We show the corresponding configuration of the privacy guarantee as well as the derivations required for the calculation of optimal scaling parameters that satisfy the guarantee. In this class, there is a single span of infeasible space extending either to infinity or to negative infinity. The complement of this space is a single span of feasible space. Without loss in generality, we assume the constraint extends to infinity. The case of going to negative infinity can be seen as a horizontal reflection of this, in which all analysis is the same.

Since there is only a single constraint, the normalization factor can be calculated as the reciprocal of the integral of the truncated space. We apply this normalization factor to the privacy guarantee of Formula \eqref{eq:guar_simp} along with the distance generalization of Section \ref{sec:gen_dist} and the continuous random variable worst-case in the third case of Formula \eqref{eq:cont_var_placement}. The privacy guarantee can now be written as:

{
	\begin{equation} \label{eq:sing_inf_guar_simp}
	\frac{\sigma_2}{\sigma_1} 
	\left( \frac{2-e^{-\frac{\Delta L_1 + i \Delta F}{\sigma_2}}}{2-e^{-\frac{\Delta L_1}{\sigma_1}}} \right) 
	\left( \frac{e^{\frac{\Delta L_1 + i \Delta F}{\sigma_2}}}{e^{\frac{\Delta L_1}{\sigma_1}}} \right)
	\leq e^{i \epsilon}
	\end{equation}
}

The formulae in this section make use of the multi-valued LambertW function \cite{14}. As we are interested only in real-valued output, we restrict our attention to the single-valued functions of the primary and -1 branches. We refer to these functions as $W$ and $W_{-1}$ for shorthand and use $W_Z$ for instances where either branch may be applied.

\begin{definition}
	The \textbf{LambertW function} \cite{14} is given by the inverse function of $f(x) = xe^x$ where $x$ is a complex number. For the branches to be single-valued, the following conditions apply: $W(x) \geq -1$ and $W_{-1}(x) \leq -1$.
\end{definition}

{
	\begin{equation} \label{eq:lambert_0}
	W(x e^x) = x \hspace{1cm} -\frac{1}{e} \leq x e^x
	\end{equation}
}

{
	\begin{equation} \label{eq:lambert_-1}
	W_{-1}(x e^x) = x \hspace{1cm} -\frac{1}{e} \leq x e^x < 0
	\end{equation}
}

\begin{lemma} \label{lem:1}
	For any PDF with a location parameter at distance $i \Delta F > 0$ from the constraint, bounds on the possible values of its scaling parameter are determined by the following four inequalities:
	
	{
		\begin{equation} \label{eq:s2_isolation_bound_1}
		\sigma_2 \leq -\frac{i \Delta F e^{i \epsilon} \sigma_1}{W \left( -\frac{2 i \Delta F e^{-i \epsilon} e^{-\frac{i \Delta F e^{-i \epsilon}}{\sigma_1}} }{\sigma_1} \right) e^{i \epsilon} \sigma_1 + i \Delta F}
		\end{equation}
	}

	{
		\begin{equation} \label{eq:s2_isolation_bound_2}
		\sigma_2 \geq -\frac{i \Delta F e^{i \epsilon} \sigma_1}{W_{-1} \left( -\frac{2 i \Delta F e^{-i \epsilon} e^{-\frac{i \Delta F e^{-i \epsilon}}{\sigma_1}} }{\sigma_1} \right) e^{i \epsilon} \sigma_1 + i \Delta F}
		\end{equation}
	}

	{
		\begin{equation} \label{eq:s2_isolation_bound_3}
		\sigma_2 \geq \frac{i \Delta F}
		{ W \left( \frac{-1}{2 e} \right) i \epsilon e^{1 -i \epsilon}
			- W \left(2 W \left( \frac{-1}{2 e} \right) i \epsilon
			e^{ W \left( \frac{-1}{2 e} \right) i \epsilon e^{1 -i \epsilon} - i \epsilon + 1} \right)}
		\end{equation}
	}

	{
		\begin{equation} \label{eq:s2_isolation_bound_4}
		\sigma_2 \leq \frac{i \Delta F}
		{ W \left( \frac{-1}{2 e} \right) i \epsilon e^{1 -i \epsilon}
			- W_{-1} \left(2 W \left( \frac{-1}{2 e} \right) i \epsilon
			e^{ W \left( \frac{-1}{2 e} \right) i \epsilon e^{1 -i \epsilon} - i \epsilon + 1} \right)}
		\end{equation}
	}
\end{lemma}

\begin{proof}
	By isolating $\sigma_2$ in Formula \eqref{eq:sing_inf_guar_simp}, we obtain an inequality which contains the LambertW function. Since the LambertW function has two branches, this gives two restrictions on the possible values that the $\sigma_2$ values can take on. By setting $\Delta L_1 = 0$, we determine that $\sigma_2$ must fall in the intersection of the two spans specified in Formulae \eqref{eq:s2_isolation_bound_1} and \eqref{eq:s2_isolation_bound_2} in order for the privacy guarantee to be satisfied for a pair of PDFs with location parameters at distances of 0 and $i \Delta F > 0$ away from the constraint. Since the LambertW function does not produce real-valued output for any input values less than $-\frac{1}{e}$, we are also able to derive restrictions on the allowable $\sigma_2$ values that produce valid input for the LambertW function. In this way, we additionally determine that the $\sigma_2$ value for a PDF with a location parameter at distance $i \Delta F > 0$ away from the constraint must fall into one of the two spans specified by Formulae \eqref{eq:s2_isolation_bound_3} and \eqref{eq:s2_isolation_bound_4}. Note that these conditions are necessary but not sufficient for satisfying the privacy guarantee.
\end{proof}

\begin{lemma} \label{lem:2}
	Through the selection of an appropriate $\sigma_1$ value when $\Delta L_1 = 0$, it is possible to calculate $\sigma_2$ values for any PDF with a location parameter at distance $i \Delta F > 0$ away from the constraint such that the inequalities of Lemma \ref{lem:1} are satisfied.
\end{lemma}

\begin{proof}
	By calculating $\sigma_1$ as shown in Formula \eqref{eq:sigma_1_inequality_2} when $\Delta L_1 = 0$, the right-hand sides of Formulae \eqref{eq:s2_isolation_bound_1} and \eqref{eq:s2_isolation_bound_2} become equal to the right-hand sides of Formulae \eqref{eq:s2_isolation_bound_3} and \eqref{eq:s2_isolation_bound_4}, respectively. As a result, in order to satisfy all four inequalities of Lemma \ref{lem:1}, a PDF with a location parameter at distance $i \Delta F > 0$ from the constraint can be assigned one of the two possible values (by setting Z to either 0 or -1) calculated using Formula \eqref{eq:s2_equality}. 
	
	{
		\begin{equation} \label{eq:sigma_1_inequality_2}
		\sigma_1 = -\frac{\Delta F}{W \left( -\frac{1}{2e} \right) e \epsilon}
		\end{equation}
	}
	
	{
		\begin{equation} \label{eq:s2_equality}
		\sigma_2 = -\frac{i \Delta F e^{i \epsilon} \sigma_1}{W_Z \left( -\frac{2 i \Delta F e^{-i \epsilon} e^{-\frac{i \Delta F e^{-i \epsilon}}{\sigma_1}} }{\sigma_1} \right) e^{i \epsilon} \sigma_1 + i \Delta F}
		\end{equation}
	}
\end{proof}

\begin{lemma} \label{lem:3}
	The sign of the denominator in the derivative taken with respect to $i$ of the $\sigma_2$ calculation of Lemma \ref{lem:2} depends on which branch is indicated by the branch index variable $Z$.
\end{lemma}

\begin{proof}
	Through analysis of the signs of all factors in the denominator, we show that the sign will always be negative for the principal branch and will always be positive for the -1 branch.
\end{proof}

\begin{lemma} \label{lem:4}
	As a function of $i$, the $\sigma_2$ calculation of Lemma \ref{lem:2} is unimodal for either branch index.
\end{lemma}

\begin{proof}
	Since the sign of the denominator of its derivative cannot change when restricted to the same branch index, a sign change could only be induced by the numerator. Through analysis of the factors in the numerator, we show that exactly one sign change occurs for each branch index. Since the full derivative has only one sign change when considering $W$ and $W_{-1}$ separately, each of the functions are unimodal.
\end{proof}

\begin{lemma} \label{lem:5}
	As a function of $i$, the input to $W_Z$ in the $\sigma_2$ calculation of Lemma \ref{lem:2} is unimodal and is initially decreasing.
\end{lemma}

\begin{proof}
	Through analysis of the input, we can show that it is equal to 0 when $i=0$ and is initially decreasing. Furthermore, we show that the derivative has exactly one sign change at $i=\frac{1}{\epsilon}$, making the function unimodal.
\end{proof}

\begin{lemma} \label{lem:6}
	As a function of $i$, the mode of the $\sigma_2$ calculation of Lemma \ref{lem:2} is a minimum for the principal branch and a maximum for the -1 branch.
\end{lemma}

\begin{proof}
	From Lemma \ref{lem:3}, we know that the sign of the denominator in the derivative of Formula \ref{eq:s2_equality} with respect to $i$  will always be negative for the principal branch and will always be positive for the -1 branch. From Lemma \ref{lem:4}, we know that a single sign change occurs in the numerator. Through further analysis of the numerator, we show that the numerator is initially positive until its zero, from which point it is negative. This results in the mode being a minimum for the principal branch and a maximum for the -1 branch.
\end{proof}

\begin{theorem} \label{th:1}
	The $\sigma$ calculation method of Lemma \ref{lem:2} provides a solution which optimally satisfies the differential privacy guarantee.
\end{theorem}

\begin{proof}
	To show that the worst-case analysis of Section \ref{sec:cont_var} holds, we show that the $\sigma$ values are monotonically decreasing as $i$ increases. Since the privacy guarantee of Formula \ref{eq:sing_inf_guar_simp} is non-symmetric, we also apply the $\sigma$ calculations to an alternate formulation which represents the symmetric case and show that this formulation holds. Finally, we must consider the optimality of the $\sigma$ values. Since lower $\sigma$ values are preferable, we show through further examination of the four inequalities of Lemma \ref{lem:1} that it is not possible to calculate lower $\sigma$ values that could satisfy the privacy guarantee.
\end{proof}

\subsection{Arbitrary Finite Constraints} \label{sec:arbitrary_finite}

In this configuration, we have an arbitrary set of finite spans of infeasible space. The complement of these spans is the feasible space, thus it extends to both positive and negative infinity.

As stated in Section \ref{sec:cont_var}, since we have no constraints that go to infinity, the continuous random variable worst-case analysis requires that we use the same scaling parameter everywhere. The normalization factors that we use here correspond to summations using finite integrals. We start by considering pairs of PDFs with location parameters in the same span of truncated space. Since the scaling parameters are all the same, we can apply the identities $L_2 = L_1 e^{\frac{i \Delta F}{\sigma}}$ and $R_2 = R_1 e^{-\frac{i \Delta F}{\sigma}}$ to produce the privacy guarantee shown in Formula \eqref{eq:finite_guar}. Note that since there are no infinite constraints present, the indicator variables for such constraints inside $L_1$ and $R_1$ are equal to 0.

{
	\begin{equation} \label{eq:finite_guar}
	\frac{1 - \left( 
		L_1 e^{\frac{i \Delta F}{\sigma}} + 
		R_1 e^{-\frac{i \Delta F}{\sigma}} \right) }
	{1 - \left( L_1 + R_1 \right) } 		
	\left( e^{\frac{i \Delta F}{\sigma}}  \right)
	\leq e^{i \epsilon}
	\end{equation}
}

For pairs of PDFs with location parameters in different spans of truncated space, the sets of constraints to their left and right will differ. This would require each normalization factor to be written in terms of its own sets of constraints. As we will see later, there is only one specific case of PDFs with location parameters in different spans of truncated space that we must explicitly consider in order to prove the guarantee for any pair of databases.

\begin{lemma} \label{lem:7}
	For each span of truncated space, there exists a value of $\sigma$ for which the privacy guarantee is satisfied for any pair of PDFs with location parameters within that span.
\end{lemma}

\begin{proof}
	When comparing the left and right-hand sides of the privacy guarantee for $i=0$, both sides are equal to 1. In order for the inequality of the privacy guarantee to be satisfied, it is therefore a necessary condition that the left-hand side is initially increasing at a lower rate than the right-hand side. By taking the derivative of each side and enforcing this condition as an inequality, it is possible to derive a lower bound for $\sigma$ shown in Formula \eqref{eq:finite_der_req_2}. By using this lower bound as the chosen $\sigma$ value, we show that the privacy guarantee is satisfied for all pairs of PDFs with location parameters in the same span of truncated space.
	
	{
		\begin{equation} \label{eq:finite_der_req_2}
		\sigma \geq \frac{\Delta F}{\epsilon - \frac{\Delta F \left( L_1 - R_1 \right)}{\sigma \left( L_1 + R_1 - 1 \right)}}
		\end{equation}
	}
\end{proof}

\begin{lemma} \label{lem:8}
	Within each span of truncated space, the $\sigma$ value determined from Lemma \ref{lem:7} acts as a lower bound for the value of $\sigma$ required to satisfy the privacy guarantee.
\end{lemma}

\begin{proof}
	With an alternate representation of the privacy guarantee which uses a value of $\sigma$ arbitrarily larger than that which is given by the lower bound in Formula \eqref{eq:finite_der_req_2}, we show that the privacy guarantee will still be satisfied.
\end{proof}

\begin{lemma} \label{lem:9}
	For each finite constraint, there exists a value of $\sigma$ that satisfies the privacy guarantee for the pair of PDFs with location parameters on the endpoints of the constraint.
\end{lemma}

\begin{proof}
	The guarantee form used thus far has used the same sets of finite constraints to the left and right of both location parameters, implying that they must both lie within the same span of truncated space. We must also be able to show that the privacy guarantee holds for location parameters in different spans of truncated space. To show this, we first consider a pair of PDFs with location parameters that lie on opposite endpoints of a finite constraint (with $f(D_1)$ as always being the point on the right). We provide this instantiation of the privacy guarantee in Formula \ref{eq:finite_guar_sep_2}, using Formula \ref{eq:c_val} to calculate the value of the integral corresponding to the span covered by the finite constraint. Since both location parameters are adjacent to this span, the value of the integral is the same for both PDFs.
	
	{
		\begin{equation} \label{eq:c_val}
		S = \frac{1-e^{-\frac{i \Delta F}{\sigma}}}{2}
		\end{equation}
	}

	{
		\begin{equation} \label{eq:finite_guar_sep_2}
		\frac{1 - \left( 
			e^{\frac{i \Delta F}{\sigma}} \left( L_1 - S \right) + 
			R_1 e^{- \frac{i \Delta F}{\sigma}} + S \right) }
		{1 - \left( L_1 + R_1 \right) }  
		\left( e^{\frac{i \Delta F}{\sigma}} \right)
		\leq e^{i \epsilon}
		\end{equation}
	}
	
	We know from Lemma \ref{lem:7} that a sufficiently high value of $\sigma$ can satisfy the guarantee without the modification made here and from Lemma \ref{lem:8} that raising the value of $\sigma$ beyond the requirement will not cause that form of the privacy guarantee to be violated. We show that increasing $\sigma$ reduces the influence of $S$ in Formula \ref{eq:finite_guar_sep_2} by causing the value of $S$ to asymptotically approach 0. It therefore follows that a sufficiently high value of $\sigma$ will also satisfy this form of the privacy guarantee and that increasing $\sigma$ beyond that value will not violate the guarantee.
\end{proof}

\begin{lemma} \label{lem:10}
	All lower bounds on $\sigma$ identified in Lemmas \ref{lem:7} and \ref{lem:9} are less than $\frac{2 \Delta F}{\epsilon}$.
\end{lemma}

\begin{proof}
	The calculations for lower bounds on $\sigma$ can be handled by treating Formulae \eqref{eq:finite_der_req_2} and \eqref{eq:finite_guar_sep_2} as equalities and solving for $\sigma$. We can identify bounds on the possible values of $\sigma$ by studying the bounds on the variables $L_1$ and $R_1$. Since $L_1$ represents the sum of the integrals of the constraints to the left of $f(D_1)$ it can be as low as 0 (if no constraints are present to the left of $f(D_1)$) and can approach but not reach 0.5 (since half of the integral exists on the left hand side). The bounds on $R_1$ are characterized in the same way. By studying the ranges of the calculations for $\sigma$ in the context of these bounds, we show $\sigma$ will fall in the range of $\left( 0,\frac{2 \Delta F}{\epsilon} \right)$.
\end{proof}

\begin{theorem} \label{th:2}
	The optimal $\sigma$ value that satisfies the privacy guarantee for all pairs of databases can be found by taking the maximum out of $3n+2$ lower bounds, where $n$ is the number of finite constraints.
\end{theorem}

\begin{proof}
	Lemma \ref{lem:7} provides a lower bound on $\sigma$ for pairs of PDFs with location parameters in the same span of truncated space. This applies to a form of the privacy guarantee in which $f(D_2) \leq f(D_1)$ holds. The symmetric case can be handled in the same way after an application of horizontal reflection to the configuration. Lemma \ref{lem:9} provides an additional lower bound on $\sigma$ for PDFs with location parameters on opposite endpoints of a constraint. In Formula \eqref{eq:finite_guar_sep_2}, if the fraction on the left-hand side is inversed, as it would be in a symmetric form, the left-hand side decreases. Thus, if the form in Formula \eqref{eq:finite_guar_sep_2} is satisfied, the symmetric form will be as well. For $n$ constraints, there are $n+1$ lower bounds on $\sigma$ in Lemma \ref{lem:7} from the regular form of the privacy guarantee and an additional $n+1$ lower bounds for the symmetric form. From Lemma \ref{lem:9}, there are $n$ lower bounds, giving a total of $3n+2$. By selecting the largest of these, the guarantee is satisfied for all pairs of PDFs with location parameters in any same span of truncated space and for all pairs of PDFs with location parameters on opposite endpoints of a constraint. Since each of the lower bounds must be adhered to, it is not possible to select a lower value of $\sigma$ than this.
	
	It remains to be shown that any arbitrary pair of databases is also protected. This follows as a transitive property of multiple applications of the guarantee forms used throughout the lemmas. For any pair of PDFs with location parameters at arbitrary points in truncated space, it is possible to represent this as a sequence of points where each adjacent pair of points in the sequence corresponds to a pair of PDF location parameters in the configuration used in either Lemma \ref{lem:7} or Lemma \ref{lem:9}. The multiplicative bound for the arbitrary pair is therefore the product of the bounds of the adjacent pairs in the sequence. Since each of the adjacent pairs satisfy the privacy guarantee, the product will also satisfy the guarantee for the arbitrary pair.
\end{proof}

\begin{theorem} \label{th:3}
	The optimal value of $\sigma$ for any configuration of $n$ arbitrary finite constraints can be calculated to a precision of $d$ decimal places in $O \left( n^2 \left( d + \log \left( \frac{\Delta F}{\epsilon} \right) \right) \right)$ time.
\end{theorem}

\begin{proof}
	As stated in Theorem \ref{th:2}, $\sigma$ must be chosen as the maximum value out of the  $O \left( n \right)$ lower bounds calculated from Formulae \eqref{eq:finite_der_req_2} and \eqref{eq:finite_guar_sep_2}. The variables $L_1$ and $R_1$ in these inequalities represent summations of $O \left( n \right)$ exponential functions where each function contains an instance of $\sigma$ in the denominator of its exponent. We know of no method to isolate $\sigma$ for such configurations. It therefore takes $O \left( n^2 \right)$ time to check whether a given value of $\sigma$ is above all lower bounds.
			
	From Lemmas \ref{lem:8} and \ref{lem:9}, we know that any value of $\sigma$ larger than the required value will also satisfy the privacy guarantee. This implies that once the inequality is satisfied, increasing $\sigma$ further will never violate the inequality. Lemma \ref{lem:10} indicates that the value of $\sigma$ will always be between 0 and $\frac{2 \Delta F}{\epsilon}$, meaning that for a decimal precision of $d$, there are  $\frac{ 10^d \left( 2 \right) \Delta F}{\epsilon}$ possible values for $\sigma$. By performing a binary search for the optimal value, the logarithm of the number of possible values of $\sigma$ must be checked, leading to an overall time complexity of $O \left( n^2 \left( d + \log \left( \frac{\Delta F}{\epsilon} \right) \right) \right)$. In most cases, the values of $d$ and $\Delta F$ are likely to be small, making $O \left( n^2 \right)$ a more practical representation of the time complexity.
\end{proof}

\subsection{Arbitrary Constraints}

While the privacy guarantee forms in the configurations covered thus far have been manageable in their complexity, configurations with multiple constraints where at least one spans infinitely lead to much more complex expressions due to an increase in the number of exponential functions which appear in the normalization factors combined with variable scaling parameters. The expressions that must be worked with now take the form of transcendental functions having differing polynomials as coefficients and exponents. Exact calculations necessary for determining optimal scaling parameter values now become very difficult, if not impossible, due to the necessary step of root-finding for these functions. Unlike the case of the single infinite constraint, there is no function such as the LambertW that can be easily substituted in these expressions to allow for solutions to be calculated.

When the number of constraints is small, it may still be possible to apply some form of analysis to produce very good approximate solutions which can approach the optimal values within a desired level of precision. We leave the analysis of such configurations as an open problem. For cases with greater numbers of constraints, it is likely that the analysis becomes unmanageable. For such cases, we thus assign the same scaling parameters for every PDF instead of considering variable scaling parameters. In essence, we sacrifice the better utility obtained by the calculation of optimal scaling parameters for the ability to calculate the optimal value when using the same scaling parameter for all PDFs such that we are able to satisfy the privacy guarantee. The corresponding privacy guarantee is shown in Formula \eqref{eq:arbitrary_guar}. This is a modified version of Formula \eqref{eq:guar_dist} using same scaling parameters and the appropriate continuous random variable worst-case analysis of Section \ref{sec:cont_var}.

{
	\begin{equation} \label{eq:arbitrary_guar}
	\frac{1 - \left( 
		L_1 e^{\frac{i \Delta F}{\sigma}} + 
		R_1 e^{-\frac{i \Delta F}{\sigma}} \right) }
	{1 - \left( L_1 + R_1 \right) } 
	\left( e^{\frac{i \Delta F}{\sigma}} \right)
	\leq e^{i \epsilon}
	\end{equation}
}

\begin{corollary}
	The optimal uniform value of $\sigma$ for any configuration of $n$ arbitrary constraints can be calculated to a precision of $d$ decimal places in $O \left( n^2 \left( d + \log \left( \frac{\Delta F}{\epsilon} \right) \right) \right)$ time.
\end{corollary}

\begin{proof}
	The calculation of $\sigma$ and the related proofs required here are almost identical to those of Section \ref{sec:arbitrary_finite}. The only difference is that $L_1$ and $R_1$ may now contain integrals of constraints that span to negative infinity and infinity respectively. Due to this, it is now possible for $L_1$ and $R_1$ to be equal to 0.5, whereas before they could only approach this value. Note however that they cannot be simultaneously equal to 0.5 as this would imply that no feasible space exists. This implies that rather than the upper bound on $\sigma$ being one that approaches $2 \Delta F / \epsilon$, is it now equal to this value. The algorithm for the approximation of the optimal value of $\sigma$ from Theorem \ref{th:3} can thus be employed here as well.
\end{proof}

\renewcommand{\arraystretch}{1.3}
\begin{table*}
	\caption{Summary of information on $\sigma$ for each constraint configuration.}
	\begin{tabular}{ | c | c | c | p{4cm} | }
		\hline
		\textbf{Constraint Configuration} & \textbf{$\sigma$ Calculation} & \textbf{Bound(s)} & \textbf{Notes} \\ \hline
		Single Infinite Constraint & $\sigma_1$: Formula \eqref{eq:sigma_1_inequality_2} & $\sigma_1 \approx \frac{1.586 \Delta F}{\epsilon}$ & Optimal $\sigma$ values \\
		 & $\sigma_2$: Formula \eqref{eq:s2_equality} & $\sigma_2 < \frac{1.586 \Delta F}{\epsilon}$ & \\ \hline
		Arbitrary Finite Constraints & Section \ref{sec:arbitrary_finite} & $\sigma < \frac{2 \Delta F}{\epsilon}$ & Approximation of optimal $\sigma$ values \\ \hline
		Arbitrary Constraints & Section \ref{sec:arbitrary_finite} & $\sigma \leq  \frac{2 \Delta F}{\epsilon}$ & Approximation of optimal uniform $\sigma$ values \\ \hline
	\end{tabular}
	\label{table:sigma_summary}
\end{table*}

\subsection{Summary of Results}

Through the study of query range constraints, we have presented a framework for the design of mechanisms that allow for the truncation and normalization of the Laplace PDF according to the configuration of constraints. We have provided a detailed analysis of the impact of constraints on the mechanisms in our framework. An important characteristic of this framework is that it maintains the ability to apply an arbitrary discretization to the truncated range of the mechanisms.

We have derived a generalized differential privacy guarantee and have applied our method of design to various classes of constraint configurations. For each mechanism, we have proven its correctness and where applicable, we have also proven optimality with respect to the calculation of the scaling parameters. Information on the calculation of scaling parameters and their bounds is summarized in Table \ref{table:sigma_summary} for each of the constraint configuration classes that we have studied.

\section{Conclusions}

When posing queries on a sensitive database, natural range constraints of the query may be publicly known, however, most randomization mechanisms do not adhere to range constraints when adding noise to query responses. Truncation and normalization of a Laplace PDF offers a means to generate noisy responses within a specified range while improving the utility of the mechanism by inducing an increase in probability density that is greater for noisy responses nearer the true response than for responses farther away. We show that since the normalization of the PDF is a data-dependent operation, this leaks sensitive information and violates the differential privacy guarantee. We propose a method to correct this which involves determining a scaling parameter then used by the Laplace distribution. We introduce a generalization of the differential privacy guarantee which can be applied to the Laplace distribution to incorporate data-dependent normalization factors. Using our generalized guarantee, we study different classes of range constraint configurations and provide a derivation of optimal scaling parameters or approximations thereof. Using our proposed calculations, a data custodian can now apply the Laplace distribution to answer general numeric queries in a range-adherent and differentially private manner.

\bibliographystyle{plain}
\bibliography{bibliography}

\newpage
\setcounter{equation}{0}
\setcounter{lemma}{0}
\setcounter{theorem}{0}

\begin{appendix}
	
	\section*{Appendix}
	
	\begin{lemma} \label{lem:app_1}
		For any PDF with a location parameter at distance $i \Delta F > 0$ from the constraint, bounds on the possible values of its scaling parameter are determined by the following four inequalities:
		
		\medskip
		\begin{equation} \label{eq:app_s2_isolation_bound_1}
		\sigma_2 \leq -\frac{i \Delta F e^{i \epsilon} \sigma_1}{W \left( -\frac{2 i \Delta F e^{-i \epsilon} e^{-\frac{i \Delta F e^{-i \epsilon}}{\sigma_1}} }{\sigma_1} \right) e^{i \epsilon} \sigma_1 + i \Delta F}
		\end{equation}
		
		\begin{equation} \label{eq:app_s2_isolation_bound_2}
		\sigma_2 \geq -\frac{i \Delta F e^{i \epsilon} \sigma_1}{W_{-1} \left( -\frac{2 i \Delta F e^{-i \epsilon} e^{-\frac{i \Delta F e^{-i \epsilon}}{\sigma_1}} }{\sigma_1} \right) e^{i \epsilon} \sigma_1 + i \Delta F}
		\end{equation}
		
		{
			\begin{equation} \label{eq:app_s2_isolation_bound_3}
			\sigma_2 \geq \frac{i \Delta F}
			{ - W \left( 2 W \left( -\frac{1}{2 e} \right) i \epsilon
				e^{ W \left( -\frac{1}{2 e} \right) i \epsilon e^{-i \epsilon + 1} - i \epsilon + 1} \right)
				+ W \left( -\frac{1}{2 e} \right) i \epsilon e^{-i \epsilon + 1}}
			\end{equation}
		}
		
		{
			\begin{equation} \label{eq:app_s2_isolation_bound_4}
			\sigma_2 \leq \frac{i \Delta F}
			{ - W_{-1} \left( 2 W \left( -\frac{1}{2 e} \right) i \epsilon
				e^{ W \left( -\frac{1}{2 e} \right) i \epsilon e^{-i \epsilon + 1} - i \epsilon + 1} \right)
				+ W \left( -\frac{1}{2 e} \right) i \epsilon e^{-i \epsilon + 1}}
			\end{equation}
		}
	\end{lemma}
	
	\begin{proof}
		Recall from the main text that the privacy guarantee to be satisfied is as shown in Formula \eqref{eq:app_sing_inf_guar_l1}.
		
		\medskip
		{
			\begin{equation} \label{eq:app_sing_inf_guar_l1}
			\frac{\sigma_2}{\sigma_1}  
			\left( \frac{2-e^{-\frac{\Delta L_1 + i \Delta F}{\sigma_2}}}{2-e^{-\frac{\Delta L_1}{\sigma_1}}} \right)
			\left( \frac{e^{\frac{\Delta L_1 + i \Delta F}{\sigma_2}}}{e^{\frac{\Delta L_1}{\sigma_1}}} \right)
			\leq e^{i \epsilon}
			\end{equation}
		}
		\bigskip
		
		By isolating $\sigma_2$ in the privacy guarantee, we obtain Formulae \eqref{eq:app_s2_isolation_l1_leq} and \eqref{eq:app_s2_isolation_l1_geq}. Given a PDF for database $D_1  \in \mathbb{D}$ with a scaling parameter $\sigma_1$, a paired PDF for database $D_2  \in \mathbb{D}$ satisfies the privacy guarantee if its scaling parameter $\sigma_2$ falls in the intersection of the two spans given by these inequalities.
		
		\medskip
		{
			\begin{equation} \label{eq:app_s2_isolation_l1_leq}
			\sigma_2 \leq
			- \frac{2 \left( e^{\frac{\Delta L_1}{\sigma_1}} - \frac{1}{2} \right) e^{i \epsilon} \sigma_1 \left( i \Delta F + \Delta L_1 \right)}
			{-\sigma_1 \left( e^{i \epsilon} - 2 e^{\frac{i \epsilon \sigma_1 + \Delta L_1}{\sigma_1}} \right) 
				W \left( -\frac{2 \left( i \Delta F + \Delta L_1 \right) 
					e^{-\frac{2 i \epsilon \sigma_1 e^{\frac{\Delta L_1}{\sigma_1}} + e^{-i \epsilon} i \Delta F - i \epsilon \sigma_1 + e^{-i \epsilon} \Delta L_1}{\sigma_1 \left( 2 e^{\frac{\Delta L_1}{\sigma_1}} - 1 \right)}}}
				{\sigma_1 \left( 2 e^{\frac{\Delta L_1}{\sigma_1}} - 1 \right)} \right)
				+ i \Delta F + \Delta L_1}
			\end{equation}
		}
		
		{
			\begin{equation} \label{eq:app_s2_isolation_l1_geq}
			\sigma_2 \geq
			- \frac{2 \left( e^{\frac{\Delta L_1}{\sigma_1}} - \frac{1}{2} \right) e^{i \epsilon} \sigma_1 \left( i \Delta F + \Delta L_1 \right)}
			{-\sigma_1 \left( e^{i \epsilon} - 2 e^{\frac{i \epsilon \sigma_1 + \Delta L_1}{\sigma_1}} \right) 
				W_{-1} \left( -\frac{2 \left( i \Delta F + \Delta L_1 \right) 
					e^{-\frac{2 i \epsilon \sigma_1 e^{\frac{\Delta L_1}{\sigma_1}} + e^{-i \epsilon} i \Delta F - i \epsilon \sigma_1 + e^{-i \epsilon} \Delta L_1}{\sigma_1 \left( 2 e^{\frac{\Delta L_1}{\sigma_1}} - 1 \right)}}}
				{\sigma_1 \left( 2 e^{\frac{\Delta L_1}{\sigma_1}} - 1 \right)} \right)
				+ i \Delta F + \Delta L_1}
			\end{equation}
		}
		\bigskip
		
		In order for this intersection to be a real-valued range, it is necessary for the input to the LambertW functions in the inequalities to always be greater than or equal to $-\frac{1}{e}$. To ensure that this condition is met for all possible values of $i$, we first take the derivative of the input with respect to $i$ as shown in Formula \eqref{eq:app_l1_lambert_input_der}.
		
		\medskip
		{
			\begin{equation} \label{eq:app_l1_lambert_input_der}
			-\frac{2 \left( i \epsilon \Delta F + \epsilon \Delta L_1 - \Delta F \right) 
				\left( -2 \sigma_1 e^{\frac{\Delta L_1}{\sigma_1}} + \left( i \Delta F + \Delta L_1 \right) e^{-i \epsilon} + \sigma_1 \right)
				e^{-\frac{2 i \epsilon \sigma_1 e^{\frac{\Delta L_1}{\sigma_1}} + e^{-i \epsilon} i \Delta F - i \epsilon \sigma_1 + e^{-i \epsilon} \Delta L_1}
					{\sigma_1 \left( 2 e^{\frac{\Delta L_1}{\sigma_1}} -1 \right)}}}
			{\left( \sigma_1 \right)^2 \left( 2 e^{\frac{\Delta L_1}{\sigma_1}} -1 \right)^2}
			\end{equation}
		}
		\bigskip
		
		Three possible zeros for the derivative can be calculated as shown in Formulae \eqref{eq:app_l1_lambert_input_der_zero} (where the variable Z can be replaced with 0 or -1) and \eqref{eq:app_l1_lambert_input_der_zero_2}.
		
		\medskip
		{
			\begin{equation} \label{eq:app_l1_lambert_input_der_zero}
			i = -\frac{\Delta L_1 \epsilon + \Delta F W_Z \left( 
				-\frac{\epsilon \sigma_1 \left( 2 e^{\frac{\Delta L_1}{\sigma_1}} -1 \right) e^{-\frac{\Delta L_1 \epsilon}{\Delta F}} }
				{\Delta F} \right)}
			{\epsilon \Delta F}
			\end{equation}
		}
		
		{
			\begin{equation} \label{eq:app_l1_lambert_input_der_zero_2}
			i = -\frac{\Delta L_1 \epsilon - \Delta F}{\epsilon \Delta F}
			\end{equation}
		}
		\bigskip
		
		Since the substitution of the zero from Formula \eqref{eq:app_l1_lambert_input_der_zero} into the input of the LambertW function does not allow for easy steps of simplification when using -1 as the value of $Z$, we take an alternate approach to first show that the value of the input to the LambertW function is the same at both zeros defined in Formula \eqref{eq:app_l1_lambert_input_der_zero} (using 0 or -1 as the value of $Z$). Let the first zero occur at $i$ and the second occur at $j$. The equality between the values at these zeros is shown in Formulae \eqref{eq:app_second_zero_equality} - \eqref{eq:app_second_zero_equality_6}.
		
		\medskip
		{
			\begin{equation} \label{eq:app_second_zero_equality}
			-\frac{2 \left( i \Delta F + \Delta L_1 \right) 
				e^{-\frac{2 i \epsilon \sigma_1 e^{\frac{\Delta L_1}{\sigma_1}} + e^{-i \epsilon} i \Delta F - i \epsilon \sigma_1 + e^{-i \epsilon} \Delta L_1}{\sigma_1 \left( 2 e^{\frac{\Delta L_1}{\sigma_1}} - 1 \right)}}}
			{\sigma_1 \left( 2 e^{\frac{\Delta L_1}{\sigma_1}} - 1 \right)}
			=
			-\frac{2 \left( j \Delta F + \Delta L_1 \right) 
				e^{-\frac{2 j \epsilon \sigma_1 e^{\frac{\Delta L_1}{\sigma_1}} + e^{-j \epsilon} i \Delta F - j \epsilon \sigma_1 + e^{-j \epsilon} \Delta L_1}{\sigma_1 \left( 2 e^{\frac{\Delta L_1}{\sigma_1}} - 1 \right)}}}
			{\sigma_1 \left( 2 e^{\frac{\Delta L_1}{\sigma_1}} - 1 \right)}
			\end{equation}
		}
		
		{
			\begin{equation} \label{eq:app_second_zero_equality_2}
			-\frac{e^{-i \epsilon} \left( i \Delta F + \Delta L_1 \right) 
				e^{-\frac{e^{-i \epsilon} \left( i \Delta F + \Delta L_1 \right)}
					{\sigma_1 \left( 2 e^{\frac{\Delta L_1}{\sigma_1}} - 1 \right)}}}
			{\sigma_1 \left( 2 e^{\frac{\Delta L_1}{\sigma_1}} - 1 \right)}
			=
			-\frac{\left( j \Delta F + \Delta L_1 \right) 
				e^{-\frac{e^{-j \epsilon} \left( j \Delta F + \Delta L_1 \right)}
					{\sigma_1 \left( 2 e^{\frac{\Delta L_1}{\sigma_1}} - 1 \right)}}}
			{\sigma_1 \left( 2 e^{\frac{\Delta L_1}{\sigma_1}} - 1 \right)}
			\end{equation}
		}
		
		{
			\begin{equation} \label{eq:app_second_zero_equality_3}
			-\frac{e^{-i \epsilon} \left( i \Delta F + \Delta L_1 \right)}
			{\sigma_1 \left( 2 e^{\frac{\Delta L_1}{\sigma_1}} - 1 \right)}
			=
			-\frac{e^{-j \epsilon} \left( j \Delta F + \Delta L_1 \right)}
			{\sigma_1 \left( 2 e^{\frac{\Delta L_1}{\sigma_1}} - 1 \right)}
			\end{equation}
		}
		
		{
			\begin{equation} \label{eq:app_second_zero_equality_4}
			e^{-i \epsilon} \left( i \Delta F + \Delta L_1 \right)
			=
			e^{-j \epsilon} \left( j \Delta F + \Delta L_1 \right)
			\end{equation}
		}
		
		At this point, we replace $i$ and $j$ with their values as determined by Formula \eqref{eq:app_l1_lambert_input_der_zero} and then we simplify the expression. In the interest of space, we substitute the instances of the LambertW function from Formula \eqref{eq:app_l1_lambert_input_der_zero}, using $a$ for the 0 branch and $b$ for the -1 branch.
		
		\medskip
		{
			\begin{equation} \label{eq:app_second_zero_equality_5}
			e^{ \left( \frac{\Delta L_1 \epsilon + \Delta F a}{\epsilon \Delta F} \right) \epsilon}
			\left( \left( -\frac{\Delta L_1 \epsilon + \Delta F a}{\epsilon \Delta F} \right) \Delta F + \Delta L_1 \right)
			=
			e^{ \left( \frac{\Delta L_1 \epsilon + \Delta F b}{\epsilon \Delta F} \right) \epsilon}
			\left( \left( -\frac{\Delta L_1 \epsilon + \Delta F b}{\epsilon \Delta F} \right) \Delta F + \Delta L_1 \right)
			\end{equation}
		}
		
		{
			\begin{equation} \label{eq:app_second_zero_equality_6}
			a e^a = b e^b
			\end{equation}
		}
		
		Since the LambertW functions of $a$ and $b$ both have the same input, the equality of Formula \eqref{eq:app_second_zero_equality_6} is confirmed to be valid. Now by taking the zero of Formula \eqref{eq:app_l1_lambert_input_der_zero} using 0 as the value for Z, we rewrite the zero in terms of $\Delta F$ as shown in Formula \eqref{eq:app_second_zero_invalidity}
		
		\medskip
		{
			\begin{equation} \label{eq:app_second_zero_invalidity}
			\Delta F = \frac{-e^{i \epsilon} \sigma_1 + 2 \sigma_1 e^{\frac{i \epsilon \sigma_1 + \Delta L_1}{\Delta L_1}} -\Delta L_1}{i}
			\end{equation}
		}
		
		When substituting $\Delta F$ in the original LambertW input with the expression in Formula \eqref{eq:app_second_zero_invalidity}, the value of the input becomes $-\frac{2}{e}$ which is outside of the allowable range of input. Since we have shown that the value of the zero using -1 as the value of Z will be the same, both of the potential modes determined by these zeros can be ignored. The function of the input is therefore unimodal and it remains to determine whether the mode is a minimum or a maximum. By setting $i$ to 0 in the derivative, we obtain the expression shown in Formula \eqref{eq:app_derivative_sign}.
		
		\medskip
		{
			\begin{equation} \label{eq:app_derivative_sign}
			\frac{2 \left( \left( \Delta L_1 + \sigma_1 \right) e^{-\frac{\Delta L_1}{\sigma_1}} -2 \sigma_1 \right) 
				e^{\frac{2 \Delta L_1}{\sigma_1 \left( -2 + e^{-\frac{\Delta L_1}{\sigma_1}} \right) }}
				\left( - \Delta L_1 \epsilon + \Delta F \right) }
			{ \left( \sigma_1 \right)^2 \left( -2 + e^{-\frac{\Delta L_1}{\sigma_1}} \right)^2 }
			\end{equation}
		}
		
		It is clear that the denominator is always positive as both factors are squared. In the numerator, four factors are present. The first, being the constant 2, and the third, being an exponential function, must always be positive. The sign of the second and fourth factors remains to be determined. We show that the second factor is always negative by proving Formula \eqref{eq:app_derivative_sign_3}.
		
		\medskip
		{
			\begin{equation} \label{eq:app_derivative_sign_3}
			\left( \Delta L_1 + \sigma_1 \right) e^{-\frac{\Delta L_1}{\sigma_1}} \leq 2 \sigma_1
			\end{equation}
		}
		
		The derivative of the left-hand side with respect to $\Delta L_1$ is shown in Formula \eqref{eq:app_derivative_sign_4}.
		
		\medskip
		{
			\begin{equation} \label{eq:app_derivative_sign_4}
			-\frac{e^{-\frac{\Delta L_1}{\sigma_1}} \Delta L_1}{\sigma_1}
			\end{equation}
		}
		
		Since all variables are non-negative, this derivative is always negative meaning that the left-hand side is decreasing as $\Delta L_1$ is increasing. It is therefore maximized when $\Delta L_1 = 0$. When making this substitution, the inequality reduces to $1 \leq 2$, thus proving that the factor is indeed always negative.
		
		The fourth factor in the numerator of the derivative is the same expression as the numerator of the zero in Formula \eqref{eq:app_l1_lambert_input_der_zero_2}. From this, we can infer that if the zero is at a positive $i$ value, then the derivative is negative when $i=0$ and if the zero is at a negative $i$ value, then the derivative is positive when $i=0$. This means that the mode at the zero is a minimum.
		
		Now, we can define a condition on the input to the LambertW function as Formula \eqref{eq:app_sigma_1_bound}.
		
		\medskip
		{
			\begin{equation} \label{eq:app_sigma_1_bound}
			-\frac{2 \left( i \Delta F + \Delta L_1 \right) 
				e^{-\frac{2 i \epsilon \sigma_1 e^{\frac{\Delta L_1}{\sigma_1}} + e^{-i \epsilon} i \Delta F - i \epsilon \sigma_1 + e^{-i \epsilon} \Delta L_1}{\sigma_1 \left( 2 e^{\frac{\Delta L_1}{\sigma_1}} - 1 \right)}}}
			{\sigma_1 \left( 2 e^{\frac{\Delta L_1}{\sigma_1}} - 1 \right)}
			\geq -\frac{1}{e}
			\end{equation}
		}
		
		By isolating for $\sigma_1$, we get Formulae \eqref{eq:app_sigma_1_bound_2_geq} and \eqref{eq:app_sigma_1_bound_2_leq}. These inequalities represent disjoint spans where the value of $\sigma_1$ is required to fall into their union.
		
		\medskip
		{
			\begin{equation} \label{eq:app_sigma_1_bound_2_geq}
			\hspace{-1cm}
			\sigma_1 \geq \frac{-\Delta L_1 \left( i \Delta F + \Delta L_1 \right) }
			{ \left( i \Delta F + \Delta L_1 \right)
				W \left( \frac{2 \Delta L_1 W \left( -\frac{1}{2 e} \right) 
					e^{\frac{\Delta L_1 W \left( -\frac{1}{2 e} \right) e^{i \epsilon} + i \epsilon \left( i \Delta F + \Delta L_1 \right)}{i \Delta F + \Delta L_1}}}
				{i \Delta F + \Delta L_1} \right)
				- \Delta L_1 W \left( -\frac{1}{2 e} \right) e^{i \epsilon}}
			\end{equation}
		}

		{
			\begin{equation} \label{eq:app_sigma_1_bound_2_leq}
			\hspace{-1cm}
			\sigma_1 \leq \frac{-\Delta L_1 \left( i \Delta F + \Delta L_1 \right) }
			{ \left( i \Delta F + \Delta L_1 \right)
				W_{-1} \left( \frac{2 \Delta L_1 W \left( \frac{-1}{2 e} \right) 
					e^{\frac{\Delta L_1 W \left( \frac{-1}{2 e} \right) e^{i \epsilon} + i \epsilon \left( i \Delta F + \Delta L_1 \right)}{i \Delta F + \Delta L_1}}}
				{i \Delta F + \Delta L_1} \right)
				- \Delta L_1 W \left( \frac{-1}{2 e} \right) e^{i \epsilon}}
			\end{equation}
		}
		
		We substitute $i$ with the zero of the derivative since we must ensure that the inequalities will hold at the minimum value of the input function. This produces Formulae \eqref{eq:app_sigma_1_bound_3_geq} and \eqref{eq:app_sigma_1_bound_3_leq}. These inequalities represent bounds on $\sigma_1$ values for PDFs with location parameters at a distance of $\Delta L_1 > 0$ away from the constraint.
		
		\medskip
		{
			\begin{equation} \label{eq:app_sigma_1_bound_3_geq}
			\hspace{-2cm}
			\sigma_1 \geq \frac{- \Delta L_1 \Delta F}
			{ \Delta F W \left( \frac{2 \Delta L_1 W \left( \frac{-1}{2 e} \right) \epsilon
					e^{\frac{\Delta L_1 W \left( \frac{-1}{2 e} \right) \epsilon e^{-\frac{\Delta L_1 \epsilon - \Delta F}{\Delta F}} - \Delta L_1 \epsilon + \Delta F }{\Delta F}}}
				{\Delta F} \right)
				- \Delta L_1 W \left( \frac{-1}{2 e} \right) \epsilon e^{-\frac{\Delta L_1 \epsilon - \Delta F}{\Delta F}}}
			\end{equation}
		}
		
		{
			\begin{equation} \label{eq:app_sigma_1_bound_3_leq}
			\hspace{-2cm}
			\sigma_1 \leq \frac{- \Delta L_1 \Delta F}
			{ \Delta F W_{-1} \left( \frac{2 \Delta L_1 W \left( \frac{-1}{2 e} \right) \epsilon
					e^{\frac{\Delta L_1 W \left( \frac{-1}{2 e} \right) \epsilon e^{-\frac{\Delta L_1 \epsilon - \Delta F}{\Delta F}} - \Delta L_1 \epsilon + \Delta F }{\Delta F}}}
				{\Delta F} \right)
				- \Delta L_1 W \left( \frac{-1}{2 e} \right) \epsilon e^{-\frac{\Delta L_1 \epsilon - \Delta F}{\Delta F}}}
			\end{equation}
		}
		
		Finally, since we want to be able to consider the bounds of Formulae \eqref{eq:app_s2_isolation_l1_leq} and \eqref{eq:app_s2_isolation_l1_geq} as well as those of Formulae \eqref{eq:app_sigma_1_bound_3_geq} and \eqref{eq:app_sigma_1_bound_3_leq} for $\sigma$ values of PDFs having the same location parameter, we set $\Delta L_1=0$ in Formulae \eqref{eq:app_s2_isolation_l1_leq} and \eqref{eq:app_s2_isolation_l1_geq} and we set $\Delta L_1 = i \Delta F$ in Formulae \eqref{eq:app_sigma_1_bound_3_geq} and \eqref{eq:app_sigma_1_bound_3_leq}. This produces four inequalities which specify bounds on $\sigma_2$ values for PDFs with location parameters at distance $i \Delta F > 0$ from the constraint such that Formulae \eqref{eq:app_s2_isolation_bound_1} and \eqref{eq:app_s2_isolation_bound_2} are overlapping spans where $\sigma_2$ must fall in their intersection and Formulae \eqref{eq:app_s2_isolation_bound_3} and \eqref{eq:app_s2_isolation_bound_4} are disjoint spans where $\sigma_2$ must fall in their union.
	\end{proof}

	\begin{lemma} \label{lem:app_2}
		Through the selection of an appropriate $\sigma_1$ value when $\Delta L_1 = 0$, it is possible to calculate $\sigma_2$ values for any PDF with a location parameter at distance $i \Delta F > 0$ away from the constraint such that the inequalities of Lemma \ref{lem:app_1} are satisfied.
	\end{lemma}

	\begin{proof}
		By setting $\Delta L_1=0$ in Formula \eqref{eq:app_sing_inf_guar_l1}, we obtain the form in Formula \eqref{eq:app_sing_inf_guar_simp}.
		
		\medskip
		{
			\begin{equation} \label{eq:app_sing_inf_guar_simp}
			\frac{\sigma_2}{\sigma_1} 
			\left( 2-e^{-\frac{i \Delta F}{\sigma_2}} \right) 
			e^{\frac{i \Delta F}{\sigma_2}} 
			\leq e^{i \epsilon}
			\end{equation}
		}
		\bigskip
		
		By isolating for $\sigma_2$, we obtain Formulae \eqref{eq:app_sing_inf_sigma2_leq} and \eqref{eq:app_sing_inf_sigma2_geq}. As in Lemma \ref{lem:app_1}, the intersection of the spans defined by the inequalities gives the valid range for the $\sigma_2$ values.
		
		\begin{equation} \label{eq:app_sing_inf_sigma2_leq}
		\sigma_2 \leq -\frac{i \Delta F e^{i \epsilon} \sigma_1}{W \left( -\frac{2 i \Delta F e^{-i \epsilon} e^{-\frac{i \Delta F e^{-i \epsilon}}{\sigma_1}} }{\sigma_1} \right) e^{i \epsilon} \sigma_1 + i \Delta F}
		\end{equation}
		
		\begin{equation} \label{eq:app_sing_inf_sigma2_geq}
		\sigma_2 \geq -\frac{i \Delta F e^{i \epsilon} \sigma_1}{W_{-1} \left( -\frac{2 i \Delta F e^{-i \epsilon} e^{-\frac{i \Delta F e^{-i \epsilon}}{\sigma_1}} }{\sigma_1} \right) e^{i \epsilon} \sigma_1 + i \Delta F}
		\end{equation}
		\bigskip
		
		We now take the same process as in Lemma \ref{lem:app_1} to ensure that the input to the LambertW function is always greater than or equal to $-\frac{1}{e}$. The derivative of the input is shown in Formula \eqref{eq:app_lambert_input_der}.
		
		\medskip
		{
			\begin{equation} \label{eq:app_lambert_input_der}
			-\frac{2 \Delta F e^{-i \epsilon} e^{-\frac{i \Delta F e^{-i \epsilon}}{\sigma_1}} \left( i \epsilon - 1 \right) \left( i \Delta F e^{-i \epsilon} - \sigma_1 \right)}{{\sigma_1}^2}
			\end{equation}
		}
		\bigskip
		
		We are interested in the modality of the input, thus we calculate the zeros of the derivative. Two possible zeros are given in Formula \eqref{eq:app_lambert_input_der_zeros} where Z can be replaced with 0 or -1 and a third is given in Formula \eqref{eq:app_lambert_input_der_zeros_2}.
		
		\medskip
		{
			\begin{equation} \label{eq:app_lambert_input_der_zeros}
			i = -\frac{W_Z \left( -\frac{\epsilon \sigma_1}{\Delta F} \right) }{\epsilon}
			\end{equation}
		}
		
		{
			\begin{equation} \label{eq:app_lambert_input_der_zeros_2}
			i = \frac{1}{\epsilon}
			\end{equation}
		}
		
		For both of the first two zeros in Formula \eqref{eq:app_lambert_input_der_zeros}, by substituting $i$ with this value in the input of the LambertW function, the value of the input becomes $-\frac{2}{e}$. Since this is outside of the valid range for the input, both of these zeros can be ignored. The third zero allows for the input to fall in the valid range. We can therefore conclude that the input is unimodal as a function of $i$.
		
		We now turn to identifying whether the zero is a minimum or a maximum. To do so, we examine the value of the derivative when $i=0$. Since the value is $-\frac{2 \Delta F}{\sigma_1}$, and all variables must be non-negative, the function is decreasing at $i=0$ and since the mode is at $i=\frac{1}{\epsilon}$, it is a minimum. From this, we can now specify that in order for the input of the LambertW to always remain greater than or equal to $-\frac{1}{e}$, the value of the input at its mode must satisfy the same condition. This is represented as Formula \eqref{eq:app_sigma_1_inequality}.
		
		\medskip
		{
			\begin{equation} \label{eq:app_sigma_1_inequality}
			-\frac{2 \Delta F e^{-1-\frac{\Delta F}{\epsilon e \sigma_1}}}{\epsilon \sigma_1} \geq -\frac{1}{e}
			\end{equation}
		}
		\bigskip
		
		Using this inequality, we can now isolate for $\sigma_1$ and find that its possible values are given as the union of two disjoint spans defined by Formulae \eqref{eq:app_sigma_1_inequality_2} and \eqref{eq:app_sigma_1_inequality_3}.
		
		\medskip
		{
			\begin{equation} \label{eq:app_sigma_1_inequality_2}
			\sigma_1 \geq -\frac{\Delta F}{W \left( -\frac{1}{2e} \right) e \epsilon}
			\end{equation}
		}
		
		{
			\begin{equation} \label{eq:app_sigma_1_inequality_3}
			\sigma_1 \leq -\frac{\Delta F}{W_{-1} \left( -\frac{1}{2e} \right) e \epsilon}
			\end{equation}
		}
		\bigskip
		
		From these two spans, we select the lower bound of Formula \eqref{eq:app_sigma_1_inequality_2} as the value to assign to $\sigma_1$. Over the course of the remaining lemmas, we will show why this choice is necessary. By substituting the $\sigma_1$ of Formulae \eqref{eq:app_s2_isolation_bound_1} and \eqref{eq:app_s2_isolation_bound_2} with this value, we get the following inequalities:
		
		\medskip
		{
			\begin{equation} \label{eq:app_s2_isolation_bound_5}
			\sigma_2 \leq -\frac{i \Delta F}
			{ - W \left( 2 W \left( -\frac{1}{2 e} \right) i \epsilon
				e^{ W \left( -\frac{1}{2 e} \right) i \epsilon e^{-i \epsilon + 1} - i \epsilon + 1} \right)
				+ W \left( -\frac{1}{2 e} \right) i \epsilon e^{-i \epsilon + 1}}
			\end{equation}
		}
		
		{
			\begin{equation} \label{eq:app_s2_isolation_bound_6}
			\sigma_2 \geq -\frac{i \Delta F}
			{ - W_{-1} \left( 2 W \left( -\frac{1}{2 e} \right) i \epsilon
				e^{ W \left( -\frac{1}{2 e} \right) i \epsilon e^{-i \epsilon + 1} - i \epsilon + 1} \right)
				+ W \left( -\frac{1}{2 e} \right) i \epsilon e^{-i \epsilon + 1}}
			\end{equation}
		}
		\bigskip
		
		Note that these inequalities are identical to Formulae \eqref{eq:app_s2_isolation_bound_3} and \eqref{eq:app_s2_isolation_bound_4} except that the directions of the inequality signs are flipped. As a result, by using Formulae \eqref{eq:app_s2_isolation_bound_3} and \eqref{eq:app_s2_isolation_bound_4} as equalities rather than inequalities, for any distance $i \Delta F > 0$, we can calculate exactly two possible $\sigma_2$ values which satisfy all four inequalities of Lemma \ref{lem:app_1} when using this choice of $\sigma_1$. The calculation of $\sigma_2$ (with $\sigma_1$ left as a variable) is shown in Formula \eqref{eq:app_s2_equality}.
		
		\medskip
		\begin{equation} \label{eq:app_s2_equality}
		\sigma_2 = -\frac{i \Delta F e^{i \epsilon} \sigma_1}{W_Z \left( -\frac{2 i \Delta F e^{-i \epsilon} e^{-\frac{i \Delta F e^{-i \epsilon}}{\sigma_1}} }{\sigma_1} \right) e^{i \epsilon} \sigma_1 + i \Delta F}
		\end{equation}
		\bigskip
	\end{proof}

	\begin{lemma} \label{lem:app_3}
		The sign of the denominator in the derivative taken with respect to $i$ of the $\sigma_2$ calculation of Lemma \ref{lem:app_2} depends on which branch is indicated by the branch index variable $Z$.
	\end{lemma}
	
	\begin{proof}
		The derivative for the calculation of $\sigma_2$ is shown in Formula \eqref{eq:app_sing_inf_sigma2_der}.
		
		\medskip
		\begin{equation} \label{eq:app_sing_inf_sigma2_der}
		-\frac{\left( W_Z \left( -\frac{2 i \Delta F e^{-i \epsilon} e^{-\frac{i \Delta F e^{-i \epsilon}}{\sigma_1}} }{\sigma_1} \right) + i \epsilon \right) \Delta F e^{i \epsilon} \sigma_1}{\left( W_Z \left( -\frac{2 i \Delta F e^{-i \epsilon} e^{-\frac{i \Delta F e^{-i \epsilon}}{\sigma_1}} }{\sigma_1} \right)  e^{i \epsilon} \sigma_1 + i \Delta F \right) \left( W_Z \left( -\frac{2 i \Delta F e^{-i \epsilon} e^{-\frac{i \Delta F e^{-i \epsilon}}{\sigma_1}} }{\sigma_1} \right) + 1\right)}
		\end{equation}
		\bigskip
		
		We start by studying the sign of the first factor in the denominator. We will show that it is always negative by proving Formula \eqref{eq:app_lemma_3_step_1}.
		
		\medskip
		{
			\begin{equation} \label{eq:app_lemma_3_step_1}
			W_Z \left( -\frac{2 i \Delta F e^{-i \epsilon} e^{-\frac{i \Delta F e^{-i \epsilon}}{\sigma_1}} }{\sigma_1} \right)  e^{i \epsilon} \sigma_1 < -i \Delta F
			\end{equation}
		}
		\bigskip
		
		This can be re-written as:
		
		\medskip
		{
			\begin{equation} \label{eq:app_lemma_3_step_2}
			W_Z \left( -\frac{2 i \Delta F e^{-i \epsilon} e^{-\frac{i \Delta F e^{-i \epsilon}}{\sigma_1}} }{\sigma_1} \right) < -\frac{i \Delta F e^{-i \epsilon}}{\sigma_1}
			\end{equation}
		}
		\bigskip
		
		Since $W_Z(x e^x) = x$, we can replace the right-hand side with the appropriate LambertW function. To do this, we must determine which branch should be used. If the value of $x$ in the aforementioned identity is greater than or equal to -1, the principal branch should be used, otherwise, the -1 branch should be used. To determine which of these should be used, we substitute the $\sigma_1$ variable of the right-hand side with its selected value from Lemma \ref{lem:app_2} and analyze the resulting function of the variable $i$ shown in Formula \eqref{eq:app_lemma_3_step_3_a}.

		\medskip
		{
			\begin{equation} \label{eq:app_lemma_3_step_3_a}
			i \epsilon e^{-i \epsilon + 1} W \left( -\frac{1}{2e} \right)
			\end{equation}
		}
		\bigskip
		
		To identify what values this function can take on, we first calculate its derivate as shown in Formula \eqref{eq:app_lemma_3_step_3_b}.
		
		\medskip
		{
			\begin{equation} \label{eq:app_lemma_3_step_3_b}
			\epsilon e^{- i \epsilon + 1} W \left( -\frac{1}{2e} \right) \left( 1 - i \epsilon \right)
			\end{equation}
		}
		\bigskip
		
		Since the derivative has a single zero at $i=\frac{1}{\epsilon}$, the function is unimodal. The value of the derivative at $i=0$ is negative therefore, the mode is a minimum. The lowest value that Formula \eqref{eq:app_lemma_3_step_3_a} can take on is therefore $W \left( -\frac{1}{2e} \right)$ which is greater than -1. As a result, the principal branch should be used to apply the LambertW identity to Formula \eqref{eq:app_lemma_3_step_2}. This produces Formula \eqref{eq:app_lemma_3_step_4}.
		
		\medskip
		{
			\begin{equation} \label{eq:app_lemma_3_step_4}
			W_Z \left( -\frac{2 i \Delta F e^{-i \epsilon} e^{-\frac{i \Delta F e^{-i \epsilon}}{\sigma_1}} }{\sigma_1} \right) < W \left( -\frac{i \Delta F e^{-i \epsilon} e^{-\frac{i \Delta F e^{-i \epsilon}}{\sigma_1}}}{\sigma_1} \right)
			\end{equation}
		}
		\bigskip
		
		To determine the relationship between the left-hand side and right-hand side, we can compare the input being given to each function:
		
		\medskip
		{
			\begin{equation} \label{eq:app_lemma_3_step_5}
			-\frac{2 i \Delta F e^{-i \epsilon} e^{-\frac{i \Delta F e^{-i \epsilon}}{\sigma_1}} }{\sigma_1} < -\frac{i \Delta F e^{-i \epsilon} e^{-\frac{i \Delta F e^{-i \epsilon}}{\sigma_1}}}{\sigma_1}
			\end{equation}
		}
		\bigskip
		
		This relationship can be simplified to $-2 < -1$. This shows that the input given to the function on the left-hand side will always be less than the input given to the function on the right-hand side. The relationship between the two LambertW functions can now be determined based on which branch the variable $Z$ indicates. If $Z$ is 0 then the left-hand side is less than the right-hand side, making the factor negative. If $Z$ is -1 then the left-hand side will still be less than the right-hand side since the output of the -1 branch is always less than or equal to the output of the principal branch. Therefore, in either case, the factor is negative.
		
		Next, we study the second factor in the denominator of the derivative:
		
		\medskip
		{
			\begin{equation} \label{eq:app_lemma_3_step_6}
			W_Z \left( -\frac{2 i \Delta F e^{-i \epsilon} e^{-\frac{i \Delta F e^{-i \epsilon}}{\sigma_1}} }{\sigma_1} \right) + 1
			\end{equation}
		}
		\bigskip
		
		Since the output of the principal branch of a LambertW function is always greater than or equal to -1, this factor will always be positive when $Z$ is 0. If $Z$ is -1 then the output of the LambertW function is always less than or equal to -1, making the sign of the factor negative.
		
		Taking the signs of both factors into account, we can see that the sign of the denominator will always be negative for the principal branch and will always be positive for the -1 branch.
	\end{proof}

	\begin{lemma} \label{lem:app_4}
		As a function of $i$, the $\sigma_2$ calculation of Lemma \ref{lem:app_2} is unimodal for either branch index.
	\end{lemma}
	
	\begin{proof}
		Since the sign of the denominator of its derivative cannot change when restricted to the same branch index, the only remaining component that could induce a sign change is the first factor in the numerator:
		
		{
			\begin{equation} \label{eq:app_lemma_4_step_1}
			W_Z \left( 
			\frac{2 i \Delta F e^{\frac{-i \Delta F}{\sigma_1 e^{i \epsilon}}} }
			{- \sigma_1 e^{i \epsilon}} 
			\right) + i \epsilon
			\end{equation}
		}
		
		To find when this can be equal to 0, we use the identity of $W_Z(xe^x) = x$, setting the variable $x$ to be $-i \epsilon$. We therefore must find when the following equality is true:
		
		{
			\begin{equation} \label{eq:app_lemma_4_step_2}
			-\frac{2 i \Delta F e^{-i \epsilon} e^{-\frac{i \Delta F e^{-i \epsilon}}{\sigma_1}} }{\sigma_1} = -i \epsilon e^{-i \epsilon}
			\end{equation}
		}
		
		When $i=0$, both sides will be equal to 0, however this does not constitute a sign change in the derivative since this is at the beginning of the input range. To find any other zeros, we can isolate $i$ to get:
		
		{
			\begin{equation} \label{eq:app_lemma_4_step_3}
			i = -\frac{W_Z \left( \frac{ln \left( \frac{\epsilon \sigma_1}{2 \Delta F} \right) \epsilon \sigma_1}{\Delta F} \right)}{\epsilon}
			\end{equation}
		}
		
		Each LambertW branch therefore has its own zero using the corresponding branch index. Since the full derivative has only one sign change when considering $W$ and $W_{-1}$ separately, each of the functions are unimodal.
	\end{proof}

	\begin{lemma} \label{lem:app_5}
		As a function of $i$, the input to $W_Z$ in the $\sigma_2$ calculation of Lemma \ref{lem:app_2} is unimodal and is initially decreasing.
	\end{lemma}

	\begin{proof}
		Recall that the input to the LambertW function is as shown in Formula \eqref{eq:app_lemma_5_step_1}.
		
		\medskip
		{
			\begin{equation} \label{eq:app_lemma_5_step_1}
			-\frac{2 i \Delta F e^{-i \epsilon} e^{-\frac{i \Delta F e^{-i \epsilon}}{\sigma_1}} }{\sigma_1}
			\end{equation}
		}
		\bigskip
		
		The input is equal to 0 when $i=0$. Since all variables used here must be greater than or equal to 0 and there is a negative sign in front of the input, the value will always be less than or equal to 0. We can therefore infer that it must initially be decreasing as $i$ increases. To determine the modality, we study the derivative of the input with respect to $i$ as shown in Formula \eqref{eq:app_lemma_5_step_2}:
		
		\medskip
		{
			\begin{equation} \label{eq:app_lemma_5_step_2}
			-\frac{2 \Delta F e^{-i \epsilon} e^{-\frac{i \Delta F e^{-i \epsilon}}{\sigma_1}} \left( i \epsilon - 1 \right) \left( i \Delta F e^{-i \epsilon} - \sigma_1 \right)}{{\sigma_1}^2}
			\end{equation}
		}
		\bigskip
		
		From the derivative, we can see that a sign change could occur when any of the following cases occur:
		
		\medskip
		{
			\begin{equation} \label{eq:app_lemma_5_step_3}
			i = \frac{1}{\epsilon}
			\end{equation}
		}
		\bigskip
		
		\medskip
		{
			\begin{equation} \label{eq:app_lemma_5_step_4}
			\sigma_1 = i \Delta F e^{-i \epsilon}
			\end{equation}
		}
		\bigskip
		
		Formula \eqref{eq:app_lemma_5_step_3} represents a valid point in the span of truncated space and thus constitutes a change in modality.
		
		For Formula \eqref{eq:app_lemma_5_step_4}, we have isolated this in terms of $\sigma_1$ in order to study the LambertW function at this point. By substituting $\sigma_1$ in Formula \eqref{eq:app_lemma_5_step_1} with the right-hand side of Formula \eqref{eq:app_lemma_5_step_4}, we obtain the following:
		
		\medskip
		{
			\begin{equation} \label{eq:app_lemma_5_step_5}
			-\frac{2 i \Delta F e^{-i \epsilon} e^{-\frac{i \Delta F e^{-i \epsilon}}{i \Delta F e^{-i \epsilon}}} }{i \Delta F e^{-i \epsilon}}
			\end{equation}
		}
		\bigskip
		
		After simplification, this expression reduces to $-\frac{2}{e}$. Since the LambertW function only produces real-valued output when the input given is in the range of $[-\frac{1}{e}, \infty)$, we can see that a sign change cannot occur here. The function of Formula \eqref{eq:app_lemma_5_step_1} is therefore unimodal with its mode occurring at the zero specified in Formula \eqref{eq:app_lemma_5_step_3}.
	\end{proof}

	\begin{lemma} \label{lem:app_6}
		As a function of $i$, the mode of the $\sigma_2$ calculation of Lemma \ref{lem:app_2} is a minimum for the principal branch and a maximum for the -1 branch.
	\end{lemma}
	
	\begin{proof}
		From Lemma \ref{lem:app_4}, we know that the function is unimodal and that the sign of its rate of change depends on the sign of the function in Formula \eqref{eq:app_lemma_4_step_1}. From Lemma \ref{lem:app_5}, we know that the input to $W_Z$ in Formula \eqref{eq:app_s2_equality} is initially decreasing until its mode and is then increasing. This implies that $W$ will also be decreasing until the same mode and then increasing and that $W_{-1}$ will be increasing until the same mode and then decreasing.
		
		Since the function of $i \epsilon$ is monotonically increasing, the only way a zero can occur in Formula \eqref{eq:app_lemma_4_step_1} for the principal branch is if $W$ is initially decreasing faster than $i \epsilon$ is increasing such that the sign of Formula \eqref{eq:app_lemma_4_step_1} is initially negative. In this way, as the rate of change of $W$ increases, the sign will eventually become positive. This implies that Formula \eqref{eq:app_lemma_4_step_1} must initially be negative for the principal branch. For the -1 branch, since the input is initially decreasing from 0, $W_{-1}$ will be initially increasing from negative infinity, making Formula \eqref{eq:app_lemma_4_step_1} initially negative.
		
		Thus, for both branches, the overall sign of the numerator of the derivative is positive prior to the mode. From Lemma \ref{lem:app_3}, we know that the sign of the denominator is negative for the principal branch, making its mode a minimum and that the sign of the denominator is positive for the -1 branch, making its mode a maximum.
	\end{proof}

	\begin{theorem} \label{th:app_1}
		The $\sigma$ calculation method of Lemma \ref{lem:app_2} provides a solution which optimally satisfies the differential privacy guarantee.
	\end{theorem}
	
	\begin{proof}
		Lemma \ref{lem:app_2}, provides a method of calculating $\sigma$ values that satisfy the bounds identified in Lemma \ref{lem:app_1}. The additional condition that the $\sigma$ values are monotonically decreasing as $i$ increases must also be proven in order for the worst-case analysis of the continuous random variable (Section 3.5, main document) to be applicable.
		
		When using the $\sigma_1$ value of Formula \eqref{eq:app_sigma_1_inequality_2}, the mode specified in Formula \eqref{eq:app_lemma_4_step_3} reduces to $i=\frac{1}{\epsilon}$ for both branches of Formula \eqref{eq:app_s2_equality}. By Lemma \ref{lem:app_6}, the principal branch is decreasing until its mode and that the -1 branch is decreasing after its mode. Thus, by using the principal branch to calculate the values of $\sigma_2$ prior to $i=\frac{1}{\epsilon}$ and the -1 branch after this point, all $\sigma$ values will be monotonically decreasing as $i$ increases.
		
		We must also consider the symmetric form of the privacy guarantee. Let $K$ be a function representing the randomization mechanism. Since the calculations from Lemma \ref{lem:app_2} make both sides of the privacy guarantee equal to each other, the guarantee written using $K$ would be as shown in Formula \eqref{eq:app_symmetric_proof}. 
		
		{
			\begin{equation} \label{eq:app_symmetric_proof}
			\Pr \left( K(D_1) = x \right) = e^{i \epsilon} \Pr \left( K(D_2) = x \right)
			\hspace{2em} \forall D_1, D_2 \in \mathbb{D} : f \left( D_1 \right) \geq f \left( D_2 \right)
			\end{equation}
		}
		
		To prove the symmetric form of the guarantee, we must show that the condition shown in Formula \eqref{eq:app_symmetric_proof_2} holds.
		
		{
			\begin{equation} \label{eq:app_symmetric_proof_2}
			\Pr \left( K(D_2) = x \right) \leq e^{i \epsilon} \Pr \left( K(D_1) = x \right)
			\hspace{2em} \forall D_1, D_2 \in \mathbb{D} : f \left( D_1 \right) \geq f \left( D_2 \right)
			\end{equation}
		}
		
		Since $\epsilon$ and $i$ must be positive, we know that $e^{i \epsilon} \geq 1$. Therefore, from Formula \eqref{eq:app_symmetric_proof}, we can infer Formula \eqref{eq:app_symmetric_proof_3} which shows that the guarantee for the symmetric form is satisfied.
		
		{
			\begin{equation} \label{eq:app_symmetric_proof_3}
			\Pr \left( K(D_2) = x \right) \leq \Pr \left( K(D_1) = x \right)
			\hspace{2em} \forall D_1, D_2 \in \mathbb{D} : f \left( D_1 \right) \geq f \left( D_2 \right)
			\end{equation}
		}
		
		Finally, we must consider the optimality of the $\sigma$ values. Since lower $\sigma$ values are preferable, we will show that it is not possible to calculate lower $\sigma$ values that could satisfy the privacy guarantee.
		
		By Lemma \ref{lem:app_1}, $\sigma$ values must fall into one of the two spans specified in Formulae \eqref{eq:app_s2_isolation_bound_3} and \eqref{eq:app_s2_isolation_bound_4}. By Lemma \ref{lem:app_2}, the right-hand sides of Formulae \eqref{eq:app_s2_isolation_bound_3} and \eqref{eq:app_s2_isolation_bound_4} are equivalent to Formula \eqref{eq:app_s2_equality} using the $\sigma_1$ value of Formula \eqref{eq:app_sigma_1_inequality_2}. Therefore, these two bounds are characterized in the same way according to Lemma \ref{lem:app_6}. This implies that that prior to $i=\frac{1}{\epsilon}$ any $\sigma$ values chosen in the span of Formula \eqref{eq:app_s2_isolation_bound_4} would be increasing as $i$ increases which violates the requirement of monotonically decreasing $\sigma$ values. We are therefore restricted to the span of Formula \eqref{eq:app_s2_isolation_bound_3} and since we always select the lowest value in this span, our method of calculating $\sigma$ values prior to $i=\frac{1}{\epsilon}$ is optimal.
		
		After $i=\frac{1}{\epsilon}$, the $\sigma$ values we select are on the upper bound of Formula \eqref{eq:app_s2_isolation_bound_4}. It is therefore possible to select lower $\sigma$ values in this span without violating the requirement of having monotonically decreasing $\sigma$ values. However, since we are also selecting $\sigma$ values that are on the bounds of Formula \eqref{eq:app_s2_isolation_bound_2}, the only way to select a lower $\sigma$ value while still satisfying the privacy guarantee is to raise the value of all prior choices. Doing so means that higher $\sigma$ values must be used at $i=0$ and $i=\frac{1}{\epsilon}$. The use of a higher $\sigma$ value at $i=0$ has the effect of shifting the mode of the right-hand side of Formula \eqref{eq:app_s2_isolation_bound_1} to the left on the $i$-axis. Since the mode used to be at $i=\frac{1}{\epsilon}$, we now end up with a $\sigma$ value at $i=\frac{1}{\epsilon}$ which is past the mode of Formula \eqref{eq:app_s2_isolation_bound_1} and is also higher than that mode. As a result, in order to satisfy the bound of Formula \eqref{eq:app_s2_isolation_bound_1}, it is necessary that some of the preceding $\sigma$ values are lower than that at $i=\frac{1}{\epsilon}$, however, this again violates the requirement of having monotonically decreasing $\sigma$ values as $i$ increases.
		
		Thus, for values of $i$ both smaller and larger than $\frac{1}{\epsilon}$, the $\sigma$ values calculated in Lemma \ref{lem:app_2} are optimal.
	\end{proof}

	\begin{lemma} \label{lem:app_7}
		For each span of truncated space, there exists a value of $\sigma$ for which the privacy guarantee is satisfied for any pair of PDFs with location parameters within that span.
	\end{lemma}
	
	\begin{proof}
		Recall that the privacy guarantee for this constraint configuration class is as shown in Formula \eqref{eq:app_finite_guar}.
		
		\medskip
		{
			\begin{equation} \label{eq:app_finite_guar}
			\frac{1 - \left( 
				L_1 e^{\frac{i \Delta F}{\sigma}} + 
				R_1 e^{-\frac{i \Delta F}{\sigma}} \right) }
			{1 - \left( L_1 + R_1 \right) }		
			\left( e^{\frac{i \Delta F}{\sigma}} \right)
			\leq e^{i \epsilon}
			\end{equation}
		}
		\bigskip
		
		When $i=0$, both sides of Formula \eqref{eq:app_finite_guar} will be equal to each other. It is therefore a necessary condition to ensure that the rate of change of the left-hand side with respect to the variable $i$ is initially less than or equal to that of the right-hand side. The derivatives with respect to $i$ of the left-hand side and right-hand side are shown in Formulae \eqref{eq:app_finite_guar_der_left} and \eqref{eq:app_finite_guar_der_right} respectively.
		
		\medskip
		{
			\begin{equation} \label{eq:app_finite_guar_der_left}
			\frac{\Delta F e^{\frac{i \Delta F}{\sigma}} \left( 2L_1 e^{\frac{i \Delta F}{\sigma}} - 1 \right) }
			{\sigma \left( L_1 + R_1 - 1 \right)}
			\end{equation}
		}
		\bigskip
		
		\medskip
		{
			\begin{equation} \label{eq:app_finite_guar_der_right}
			\epsilon e^{i \epsilon}
			\end{equation}
		}
		\bigskip
		
		Since we must ensure that the derivative of the left-hand is less than or equal to the derivative of the right-hand side when $i=0$, we set all instances of $i$ to 0 and obtain Formula \eqref{eq:app_finite_der_req}.
		
		\medskip
		{
			\begin{equation} \label{eq:app_finite_der_req}
			\frac{\Delta F \left( 2L_1 - 1 \right) }
			{\sigma \left( L_1 + R_1 - 1 \right)} 
			\leq \epsilon
			\end{equation}
		}
		\bigskip
		
		Through rearrangement and application of identities, we can produce Formula \eqref{eq:app_finite_der_req_2}.
		
		\medskip
		{
			\begin{equation} \label{eq:app_finite_der_req_2}
			\sigma \geq 
			\frac{\Delta F}{\epsilon - \frac{\Delta F \left( L_1 - R_1 \right)}{\sigma \left( L_1 + R_1 - 1 \right)}}
			\end{equation}
		}
		\bigskip
		
		We will consider the use of the lowest $\sigma$ value that satisfies this inequality. Since this remains a recursive definition of $\sigma$, we cannot use this to actually calculate that value. We will show later how to do so. For now, we will simply treat the right-hand side of this inequality as an identity of the required value of $\sigma$.
		
		Using this value of $\sigma$, we know that both sides of the guarantee will initially have an equal rate of change with respect to $i$. We must now show that this value of $\sigma$ preserves the guarantee for higher values of $i$. We substitute the $\sigma$ identity into exponential function appearing just before the inequality symbol in Formula \eqref{eq:app_finite_guar} to obtain the form shown in Formula \eqref{eq:app_finite_guar_proof}. Note that we do not substitute this value into the other occurrences of $\sigma$. As we will see, any positive value for $\sigma$ could be applied to those occurrences while still satisfying the guarantee.
		
		\medskip
		{
			\begin{equation} \label{eq:app_finite_guar_proof}
			\frac{1 - \left( 
				L_1 e^{\frac{i \Delta F}{\sigma}} + 
				R_1 e^{- \frac{i \Delta F}{\sigma}} \right) }
			{1 - \left( L_1 + R_1 \right) } 
			\left( e^{\frac{i \Delta F}{\frac{\Delta F}{\epsilon - \frac{\Delta F \left( L_1 - R_1 \right)}{\sigma \left( L_1 + R_1 - 1 \right)}}}} \right)
			\leq e^{i \epsilon}
			\end{equation}
		}
		\bigskip
		
		After simplification, this can be re-written as shown in Formula \eqref{eq:app_finite_guar_proof_2}.
		
		\medskip
		{
			\begin{equation} \label{eq:app_finite_guar_proof_2}
			\frac{1 - \left( 
				L_1 e^{\frac{i \Delta F}{\sigma}} + 
				R_1 e^{- \frac{i \Delta F}{\sigma}} \right) }
			{1 - \left( L_1 + R_1 \right) } 
			\left( e^{i \epsilon - \frac{i \Delta F \left( L_1 - R_1 \right)}{\sigma \left( L_1 + R_1 - 1 \right)}} \right)
			\leq e^{i \epsilon}
			\end{equation}
		}
		\bigskip
		
		Now by rearranging, we can write the guarantee as shown in Formula \eqref{eq:app_finite_guar_proof_3}.
		
		\medskip
		{
			\begin{equation} \label{eq:app_finite_guar_proof_3}
			\frac{1 - \left( 
				L_1 e^{\frac{i \Delta F}{\sigma}} + 
				R_1 e^{- \frac{i \Delta F}{\sigma}}  \right) }
			{1 - \left( L_1 + R_1 \right) }
			\leq e^{\frac{i \Delta F \left( L_1 - R_1 \right)}{\sigma \left( L_1 + R_1 - 1 \right)}}
			\end{equation}
		}
		\bigskip
		
		Since both sides of the original guarantee were equal to each other when $i=0$, regardless of the choice of $\sigma$, the same is true of this version. Furthermore, since we have chosen a value for $\sigma$ that ensures the rate of change is the same on both sides when $i=0$, the same is true in this version as well. It remains to show that as $i$ increases, the right-hand side will grow more quickly than the left-hand side. To show this, we look at the second derivatives of the left and right sides shown in Formulae \eqref{eq:app_finite_guar_2_der_left} and \eqref{eq:app_finite_guar_2_der_right} respectively.
		
		\medskip
		{
			\begin{equation} \label{eq:app_finite_guar_2_der_left}
			\frac{\Delta F^2 \left( L_1 e^{\frac{i \Delta F}{\sigma}} + R_1 e^{-\frac{i \Delta F}{\sigma}} \right) }
			{\sigma^2 \left( L_1 + R_1 - 1 \right) }
			\end{equation}
		}
		\bigskip
		
		\medskip
		{
			\begin{equation} \label{eq:app_finite_guar_2_der_right}
			\frac{\Delta F^2 \left( L_1 - R_1 \right) ^2}
			{\sigma^2 \left( L_1 + R_1 - 1 \right) ^2} 
			\left( e^{\frac{i \Delta F \left( L_1 - R_1 \right) }{\sigma \left( L_1 + R_1 - 1 \right) }} \right)
			\end{equation}
		}
		\bigskip
		
		Since the variables $L_1$ and $R_1$ represent the summations of the integrals of the finite constraints to the left and right respectively of the location parameter $f(D_1)$, both of them have a range of $[0, 0.5)$. When $L_1$ and $R_1$ are restricted to these ranges, the expression in Formula \eqref{eq:app_finite_guar_2_der_left} will always be negative and the expression in Formula \eqref{eq:app_finite_guar_2_der_right} will always be positive. This means that the derivative of the left-hand side is decreasing and the derivative of the right-hand side is increasing. As such, the left-hand side will always be less than or equal to the right-hand side, meaning that the guarantee is satisfied when using the $\sigma$ value taken as the lower bound from Formula \eqref{eq:app_finite_der_req_2}.
	\end{proof}

	\begin{lemma} \label{lem:app_8}
		Within each span of truncated space, the $\sigma$ value determined from Lemma \ref{lem:app_7} acts as a lower bound for the value of $\sigma$ required to satisfy the privacy guarantee.
	\end{lemma}
	
	\begin{proof}
		Starting from the value of $\sigma$ obtained by using Formula \eqref{eq:app_finite_der_req_2} as an equality (as was done in Lemma \ref{lem:app_7}), we can represent a larger value of $\sigma$ by subtracting some value $\alpha$ from the denominator such that $\alpha$ is not greater than or equal to the original value of the denominator. This form is shown in Formula \eqref{eq:app_finite_larger_sigma}
		
		\medskip
		{
			\begin{equation} \label{eq:app_finite_larger_sigma}
			\sigma \geq \frac{\Delta F}{\epsilon - \frac{\Delta F \left( L_1 - R_1 \right)}{\sigma \left( L_1 + R_1 - 1 \right)} - \alpha}
			\end{equation}
		}
		\bigskip
		
		Using this new $\sigma$ value, we can repeat the steps of Formulae \eqref{eq:app_finite_guar_proof} through \eqref{eq:app_finite_guar_proof_3} to obtain Formula \eqref{eq:app_finite_larger_sigma_proof}.
		
		\medskip
		{
			\begin{equation} \label{eq:app_finite_larger_sigma_proof}
			\frac{1 - \left( 
				L_1 e^{\frac{i \Delta F}{\sigma}} + 
				R_1 e^{- \frac{i \Delta F}{\sigma}} \right) }
			{1 - \left( L_1 + R_1 \right) }
			\leq e^{\frac{i \Delta F \left( L_1 - R_1 \right)}{\sigma \left( L_1 + R_1 - 1 \right)} + i \alpha}
			\end{equation}
		}
		\bigskip
		
		The only difference between Formula \eqref{eq:app_finite_guar_proof_3} and Formula \eqref{eq:app_finite_larger_sigma_proof} is that the exponent of the right-hand side has become larger, making the guarantee easier to satisfy. It therefore follows that any $\sigma$ value larger than that which was found to be required will also satisfy the privacy guarantee.
	\end{proof}

	\begin{lemma} \label{lem:app_9}
		For each finite constraint, there exists a value of $\sigma$ that satisfies the privacy guarantee for the pair of PDFs with location parameters on the endpoints of the constraint.
	\end{lemma}
	
	\begin{proof}
		The guarantee form used thus far has used the same sets of finite constraints to the left and right of both location parameters, implying that they must both lie within the same span of truncated space. We must also be able to show that the privacy guarantee holds for location parameters in different spans of truncated space. To show this, we first consider a pair of PDFs with location parameters that lie on opposite endpoints of a finite constraint (with $f(D_1)$ as always being the point on the right).
		
		For two such location parameters, the sets of constraints that lie to their left and right now differ by exactly one constraint span, the one that separates them. We can still indicate the distances between the location parameter $f(D_2)$ and the constraints to its left in terms of the set of constraints to the left of $f(D_1)$. However, we must omit the separating constraint from the summation in $L_1$ for the normalization factor of $D_2$. This can be done by subtracting the value of its integral from the summation. Logically, that same value would be added to the summation in $R_1$ of constraints to the right of $f(D_2)$, however, we must account for how the distance between $f(D_2)$ and the separating constraint differs from the cases with the other constraints. Normally, $f(D_2)$ is $i \Delta F$ farther from the each constraint to its right than $f(D_1)$ is but in this case, since the two location parameters are each on an endpoint of the separating constraint, they are both at a distance of 0 from the constraint. This means that the value of the separating integral is identical for both PDFs. As a result, rather than adding this value to the summation in $R_1$ for the normalization factor of $D_2$ (which would cause the integral to be scaled down), it is added independently to the total summation of integral values. Using a variable $S$ as defined in Formula \eqref{eq:app_c_val} to represent the integral of the separating constraint, this modified version of the privacy guarantee can be written as shown in Formula \eqref{eq:app_finite_guar_sep_2}.
		
		{
			\begin{equation} \label{eq:app_c_val}
			S = \frac{1-e^{-\frac{i \Delta F}{\sigma}}}{2}
			\end{equation}
		}

		{
			\begin{equation} \label{eq:app_finite_guar_sep_2}
			\frac{1 - \left( 
				e^{\frac{i \Delta F}{\sigma}} \left( L_1 - S \right) + 
				R_1 e^{- \frac{i \Delta F}{\sigma}} + S \right) }
			{1 - \left( L_1 + R_1 \right) }  
			\left( e^{\frac{i \Delta F}{\sigma}} \right)
			\leq e^{i \epsilon}
			\end{equation}
		}
		
		We know from Lemma \ref{lem:app_7} that a sufficiently high value of $\sigma$ can satisfy the guarantee without the modification made here and from Lemma \ref{lem:app_8} that raising the value of $\sigma$ beyond the requirement causes the difference between the left and right sides of the guarantee to grow more quickly. We can also see that increasing $\sigma$ reduces the influence of $S$ by causing the value of $S$ to asymptotically approach 0. It therefore follows that a sufficiently high value of $\sigma$ will also satisfy this form of the privacy guarantee.
	\end{proof}

	\begin{lemma} \label{lem:app_10}
		All lower bounds on $\sigma$ identified in Lemmas \ref{lem:app_7} and \ref{lem:app_9} are less than $\frac{2 \Delta F}{\epsilon}$.
	\end{lemma}
	
	The proof of this is provided in the appendix. The calculations for lower bounds on $\sigma$ can be handled by treating Formulae \eqref{eq:app_finite_der_req_2} and \eqref{eq:app_finite_guar_sep_2} as equalities and solving for $\sigma$. We can identify bounds on the possible values of $\sigma$ by studying the bounds on the variables $L_1$ and $R_1$. Since $L_1$ represents the sum of the integrals of the constraints to the left of $f(D_1)$ it can be as low as 0 (if no constraints are present to the left of $f(D_1)$) and can approach but not reach 0.5 (since half of the integral exists on the left hand side). The bounds on $R_1$ are characterized in the same way. By studying the the ranges of the lower bounds for $\sigma$, it can be shown that for bounds from both equations, $\sigma$ will fall in the range of $\left( 0,\frac{2 \Delta F}{\epsilon} \right)$.

	\begin{theorem} \label{th:app_2}
		The optimal $\sigma$ value that satisfies the privacy guarantee can be found by taking the maximum out of $3n+2$ lower bounds, where $n$ is the number of finite constraints.
	\end{theorem}
	
	\begin{proof}
		Lemma \ref{lem:app_7} provides a lower bound on $\sigma$ for pairs of PDFs with location parameters in the same span of truncated space. This applies to a form of the privacy guarantee in which $f(D_2) \leq f(D_1)$ holds. The symmetric case can be handled in the same way after an application of horizontal reflection to the configuration. Since the same value of $\sigma$ must be used everywhere, it is necessary to select the maximum of the lower bounds. By Lemma \ref{lem:app_8}, the privacy guarantee is still satisfied within each span of truncated space under the use of a value of $\sigma$ greater than the lower bound.
		
		Lemma \ref{lem:app_9} provides an additional bound on $\sigma$ for PDFs with location parameters on opposite endpoints of a constraint. In Formula \eqref{eq:app_finite_guar_sep_2}, if the fraction on the left-hand side is inversed, as it would be in a symmetric form, the left-hand side decreases. Thus if the form in Formula \eqref{eq:app_finite_guar_sep_2} is satisfied, the symmetric form will be as well.
		
		For $n$ constraints, there are $n+1$ lower bounds on $\sigma$ in Lemma \ref{lem:app_7} from the regular form of the privacy guarantee and an additional $n+1$ lower bounds for the symmetric form. From Lemma \ref{lem:app_9}, there are $n$ lower bounds, giving a total of $3n+2$. By selecting the largest of these, the guarantee is satisfied for all pairs of PDFs with location parameters in any same span of truncated space and for all pairs of PDFs with location parameters on opposite endpoints of a constraint. Since each of the lower bounds must be adhered to, it is not possible to select a lower value of $\sigma$ than this.
		
		It remains to be shown that any arbitrary pair of databases is also protected. This follows as a transitive property of multiple applications of the guarantee forms used throughout the lemmas. For any pair of PDFs with location parameters at arbitrary points in truncated space, it is possible to represent this as a sequence of points where each adjacent pair of points in the sequence corresponds to a pair of PDF location parameters in the configuration used in either Lemma \ref{lem:app_7} or Lemma \ref{lem:app_9}. The multiplicative bound for the arbitrary pair is therefore the product of the bounds of the adjacent pairs in the sequence. Since each of the adjacent pairs satisfy the privacy guarantee, the product will also satisfy the guarantee for the arbitrary pair.
	\end{proof}

	\begin{theorem} \label{th:app_3}
		The optimal value of $\sigma$ for any configuration of $n$ arbitrary finite constraints can be calculated to a precision of $d$ decimal places in $O \left( n^2 \left( d + \log \left( \frac{\Delta F}{\epsilon} \right) \right) \right)$ time.
	\end{theorem}

	\begin{proof}
		As stated in Theorem \ref{th:app_2}, $\sigma$ must be chosen as the maximum value out of the  $O \left( n \right)$ lower bounds calculated from Formulae \eqref{eq:app_finite_der_req_2} and \eqref{eq:app_finite_guar_sep_2}. The variables $L_1$ and $R_1$ in these inequalities represent summations of $O \left( n \right)$ exponential functions where each function contains an instance of $\sigma$ in the denominator of its exponent. We know of no method to isolate $\sigma$ for such configurations. It therefore takes $O \left( n^2 \right)$ time to check whether a given value of $\sigma$ is above all lower bounds.

		From Lemmas \ref{lem:app_8} and \ref{lem:app_9}, we know that any value of $\sigma$ larger than the required value will also satisfy the privacy guarantee. This tells us that once the inequality is satisfied, increasing $\sigma$ further will never violate the inequality. Lemma \ref{lem:app_10} indicates that the value of $\sigma$ will always be between 0 and $\frac{2 \Delta F}{\epsilon}$, meaning that for a decimal precision of $d$, there are  $\frac{2 \left( 10^d \right) \Delta F}{\epsilon}$ possible values for $\sigma$. By performing a binary search for the optimal value, the logarithm of the number of possible values of $\sigma$ must be checked, leading to an overall time complexity of $O \left( n^2 \left( d + \log \left( \frac{\Delta F}{\epsilon} \right) \right) \right)$. In most cases, the values of $d$ and $\Delta F$ are likely to be small, making $O \left( n^2 \right)$ a more practical representation of the time complexity.
	\end{proof}

\end{appendix}

\end{document}